\documentclass[11pt]{article}
\usepackage[font=scriptsize]{caption}

\usepackage{amsthm}
\usepackage{amsmath}
\usepackage{algorithm}
\usepackage[noend]{algpseudocode}
\newtheorem{theorem}{Theorem}[section]
\newtheorem{corollary}[theorem]{Corollary}
\newtheorem{lemma}[theorem]{Lemma}
\usepackage{cite}

\usepackage{graphicx}

\graphicspath{{mult_fig_arxive/}}

\theoremstyle{definition}

\usepackage{appendix}

\usepackage{algorithm}
\usepackage[noend]{algpseudocode}
\usepackage{amssymb}
\usepackage[colorinlistoftodos,prependcaption,textsize=small]{todonotes}

\usepackage{color}
\usepackage{tabulary}
\usepackage{mathtools}
\DeclarePairedDelimiter\floor{\lfloor}{\rfloor}
\usepackage{tcolorbox}

\newcommand\frho{\left(\frac{1}{\rho}\right)}

\newcommand\nfrac{n^{1+\frac{1}{\kappa}}}
\newcommand\nk{n^{1/\kappa}}
\newcommand\nko{n^{1/\kappa}-1}
\newcommand\mpi{\mathcal{P}_i}
\newcommand\mpp[1]{\mathcal{P}_{#1}}
\newcommand\congest{{\sf CONGEST}}

\newcommand\congestmo{{\sf CONGEST} model}
\newcommand\localmo{{\sf LOCAL} model}

\newcommand\ize{\floor*{\log \kappa\rho}}

\usepackage[papersize={8.5in,11in}]{geometry}

\usepackage{adjustbox}
\usepackage{hyperref}
\usepackage{cleveref}
\crefname{lemma}{Lemma}{Lemmas}
\author{Michael Elkin$^1$ and Shaked Matar$^1$}
\date{$^1$Department of Computer Science, Ben-Gurion University of the Negev, Beer-Sheva, Israel.\\
	Email: \texttt{elkinm@cs.bgu.ac.il, matars@post.bgu.ac.il}}

\title{Fast Deterministic Constructions of Linear-Size Spanners and Skeletons}

\begin{document}
	\maketitle
	
	\begin{abstract}

		In the distributed setting, the only existing constructions of \textit{sparse skeletons}, (i.e., subgraphs with $O(n)$ edges) either use randomization or large messages, or require $\Omega(D)$ time, where $D$ is the hop-diameter of the input graph $G$. We devise the first deterministic distributed algorithm in the CONGEST model (i.e., uses small messages) for constructing linear-size  skeletons in time 
		$2^{O(\sqrt{{\log n}\cdot{\log{\log n}}})}$.

		We can also compute a linear-size spanner with stretch $polylog(n)$ in low deterministic polynomial time, i.e., $O(n^\rho)$ for an arbitrarily small constant $\rho >0$, in the CONGEST model. 
		
		Yet another algorithm that we devise runs in $O({\log n})^{\kappa-1}$ time, for a parameter $\kappa=1,2,\dots,$ and constructs an $O({\log n})^{\kappa-1}$ spanner with $O(\nfrac)$ edges. 
		
		All our distributed algorithms are lightweight from the computational perspective, i.e., none of them employs any heavy computations.

	\end{abstract}
	
\newpage

\section{Introduction}

We study the problem of \textit{deterministically} constructing sparse \textit{spanners} in the distributed \congestmo. In this model, vertices of an input $n$-vertex unweighted, undirected connected graph $G=(V,E)$ host autonomous processors, which communicate with one another over edges of $G$ in synchronous rounds via short messages. Each message is allowed to contain $O(1)$ vertex IDs.\footnote{We assume that the vertices are equipped with distinct identity numbers, IDs, in the range $\{1,2,\dots,O(n)\}$. Our algorithms apply also when IDs are in a larger polynomial in $n$ range.} The \textit{running time} of a distributed algorithm is defined as the number of rounds of distributed communication.

A subgraph $G'=(V,H)$ is said to be an $\alpha-spanner$ of $G$, for a parameter $\alpha\geq 1$, if for every pair $u,v\in V$ of vertices, $d_{G'}(u,v)\leq \alpha \cdot d_G(u,v)$, where $d_G$ (respectively, $d_{G'}$) stands for the distance function in the graph $G$ (respectively, in the subgraph $G'$). The parameter $\alpha$ is called the \textit{stretch} of the spanner $G'$. 
Spanners are a focus of intensive research in the context of distributed algorithms 
\cite{Awerbuch85,PelegS89,cohen93,ElkinP01,Elkin01,ThorupZ01_distance_oracles,ElkinZ06,baswana2007simple,Elkin08a,DerbelGPV09,Pettie10,DerbelMZ10,ElkinN17,GrossmanP17,GhaffariK18}.

Existentially it is well-known that for every parameter $\kappa=1,2,3,\dots$, for every $n$-vertex graph $G=(V,E)$, a $(2\kappa-1)$-spanner with $O(\nfrac)$ edges exists \cite{Awerbuch85,PelegS89,AlthoferDDJS93}.
For $\kappa={\log n}$, the size of bound of these spanners is $O(n)$.
Such a spanner (with $O(n)$ edges) is also called a \textit{sparse skeleton} of $G$ (even if the stretch requirement does not hold).

The problem of constructing sparse skeletons in the distributed setting also has a long history. Dubhashi et al, \cite{DubhashiMPRS05} showed that in the \localmo\ (i.e., when messages of unbounded size are allowed), a sparse skeleton can be computed in deterministic $O({\log n})$ time, and also that a linear-size $O({\log n})$ spanner can be computed in randomized $O({\log^3 n})$ time, and in deterministic $2^{O(\sqrt{\log n})}$ time. 
Their algorithms place however a heavy computational burden on the processors that execute them.  
Derbel et al. \cite{DerbelGP07,DerbelG08,DerbelGPV09} devised many additional constructions of sparse spanners and skeletons, all in the \localmo.
The sparsest their spanners can achieve is $O(n\cdot {\log{\log n}})$, and the running time of their algorithm that achieves this level of sparsity is also randomized $polylog(n)$ and deterministic $2^{O(\sqrt{{\log n}})}$.

Pettie \cite{Pettie10} devised the first randomized algorithm in the \congestmo \ that constructs linear-size skeletons in time $2^{O({\log^*n})}\cdot {\log n}$. His skeleton is, in fact, also a $2^{O({\log^*n})}\cdot{\log n}$-spanner. These bounds were dramatically improved by \cite{MillerPX13} that devised a construction of  $(4\kappa-1)$-spanner with expected $O(\nfrac)$ edges, and by \cite{ElkinN17}, whose algorithm constructs $(2\kappa-1)$ spanner with expected $\nfrac$ edges,  for all $\kappa=1,2,\dots$. Both these randomized algorithms require $O(\kappa)$ time, and for $\kappa={\log n}$ they produce $O({\log n})$-spanners (and thus, skeletons) of linear size. 

Consequently the focus shifted to \textit{deterministic} algorithms in the \congestmo. Remarkably, no such algorithm for constructing \textit{linear-size skeletons} with running time $o(n)$ are known. There is an algorithm of Derbel et al. \cite{DerbelMZ10} that computes $(2\kappa-1)$-spanners of size $O(\nfrac)$ in time $O(n^{1-\frac{1}{\kappa}})$, but for linear size ($\kappa ={\log n}$), its running time is $O(n)$. There is an algorithm of Grossman and Parter \cite{GrossmanP17} that produces $(2\kappa-1)$-spanner with $O(\kappa\nfrac)$ edges in $O(2^\kappa\cdot n^{1/2-1/k})$, but the sparsest skeletons it can produce have size $\Omega(n\cdot {\log n})$. 
An algorithm of Barenboim et al. \cite{BarenboimEG15} constructs $O({\log n})^{\kappa-1}$-spanner of size $O(\nfrac)$ in time $O(\nk\cdot ({\log n})^{\kappa-1}))$. However, the sparsest spanners it can produce in sublinear time have size $\Omega(n\cdot polylog(n))$. A recent result of Ghaffari and Kuhn \cite{GhaffariK18} produces $(2\kappa-1)$-spanners with size $O(\kappa\nfrac\cdot{\log n})$ in deterministic $2^{O(\sqrt{\log n })}$ time. 
However, the sparsest spanners that this construction can produce contain $\Omega (n{\log^2 n})$ edges. 
The current authors recently devised a construction of near-additive $(1+\epsilon,\beta)$-spanners\footnote{A near additive spanner $G'=(V,E')$ for a graph $G=(V,E)$ is a subgraph of $G$ such that for every pair $u,v$ of vertices $d_{G'}(u,v)\leq (1+\epsilon)d_G(u,v)+\beta$.} with $O(\beta\nfrac)$ edges. 
For a parameter $\rho>0$, its running time is $O(\beta n^\rho)$, where 
$\beta =O\left(\frac{\log \kappa}{\rho\epsilon}\right)^{{\log \kappa}+\rho^{-1}+O(1)}$.
 By setting $\kappa={\log n}$ and 
 $\rho = \sqrt{\frac{{\log {\log n}}}{\log n}}$ one can get a spanner of size 
	$O( ({\log n})^{\log^{(3)}n})$ in 
	$2^{O(\sqrt{({\log n}){\log{\log n}}})}$ time. This is the sparsest these spanners can get.

	Hence if one wants to deterministically construct a \textit{linear-size} skeleton in the \congestmo, she has either to resort to a spanning tree construction, which would require $\Omega(D)$ time, where $D$ is the hop-diameter of the graph, or to the algorithm of Derbel et al. \cite{DerbelMZ10}, which requires $O(n)$ time.
	Moreover,  the fastest currently known algorithm  for getting skeletons with $o(n{\log^{2}n})$ edges requires $\Omega(n^{1/2})$ time \cite{GrossmanP17}.

In this paper we devise an algorithm with running time $2^{O(\sqrt{{\log n}{\log{\log n}}})}$ for constructing $O(n)$-size skeletons. More generally, for a  parameter $\kappa=1,2,\dots$, our algorithm constructs a $poly(\kappa)$-spanner in low polynomial time. (By "low polynomial time" we mean $n^\rho$, for an arbitrarily small  constant $\rho>0$.) Specifically, for any pair of parameters $\kappa=1,2,\dots,$ and $\rho\geq 1/\kappa$, 
our algorithm constructs a 
$(\kappa\rho)^{\log\frac{1}{\rho}}\cdot O\frho^{\frac{1}{\rho}+O(1)}$-spanner in 
$(\kappa\rho)^{\log\frac{1}{\rho}}\cdot O\frho^{\frac{1}{\rho}+O(1)}$ 
time.

Our second result is that we devise a $O(\log n)^{\kappa-1}$-time algorithm for constructing an  
$O(\log n)^{\kappa-1}$-spanner with $O(\nfrac)$ edges, improving the previous result by \cite{BarenboimEG15}, where an algorithm that constructs spanners with similar parameters requires $O(\nk\cdot ({\log n})^{\kappa-1})$ time.

	\subsection{Technical Overview}
	
	Our algorithm for constructing an $O({\log n})^{\kappa-1}$-spanner in $O(\nfrac)$ time bears some similarity to the algorithm of Barenboim et al. \cite{BarenboimEG15}. The latter algorithm, like its predecessor network decomposition  algorithm by Awerbuch et al. \cite{AwerbuchGLP89}, uses ruling sets to create superclusters around vertices of high degree. Vertices of low degree then insert all edges incident on them to the spanner. Then the algorithm \cite{BarenboimEG15} iterates this step on the cluster graph. Clusters of high degree become superclustered, while each low degree cluster $C$ adds to the spanner one edge $(u, v)$ such that $u \in C$, $v \in C'$ for each cluster $C'$ adjacent to $C$. (In the context of spanners, this approach was pioneered in the algorithm of \cite{ElkinP01} for constructing near-additive spanners.)
	
	To identify which clusters $C$ have many (specifically at least $\nk$) adjacent clusters, there is a pipelined convergecast conducted in parallel in all clusters. This convergecast incurs congestion of $\nk$, and thus requires $O(Rad(C)\cdot \nk)$, where $Rad(C)$ is the radius of a cluster $C$. 
	
	In our current algorithm we replace this convergecast by a local condition. In every cluster $C$ every vertex $v \in C$ checks if it has at least $\nk$ adjacent clusters. If for at least one $v \in C$ this condition holds, the cluster $C$ will be superclustered. For every cluster $C$ that is not superclustered, every vertex $v \in C$ will add to the spanner one edge for every adjacent cluster $C'$ to $v$. 
	
	This algorithm can be implemented without any congestion and thus its running time is polylogarithmic in $n$. It is not hard to see that the size of the resulting spanner is still $O(\nfrac)$. 
	To our knowledge this is the \textit{first} known deterministic algorithm that runs in polylogarithmic time in the \congestmo\ that produces sparse spanners with \textit{any meaningful} stretch guarantee. 
	
	Our algorithm for constructing $poly(\kappa)$-spanner with $O(\nfrac)$ edges in low polynomial time is also closely related to the aforementioned algorithm of Barenboim et al. \cite{BarenboimEG15}. There are a few important changes though. First, while the algorithm \cite{BarenboimEG15} uses the same degree threshold $\nk$ on all its iterations throughout the algorithm, our current algorithm employs the degree sequence $\nk, n^{2/\kappa}, n^{4/\kappa}, \ldots$ for all phases $i$ such that $\frac{2^i}{\kappa}\leq \rho$, and $n^\rho$ for all remaining phases, i.e., phases $i$ such that $\frac{2^i}{\kappa} > \rho$. (Each degree threshold determines which clusters will be superclustered on the current iteration and which, low degree, clusters will add edges to the spanner.) This degree sequence was used by Elkin and Neiman \cite{ElkinN17} and by the current authors \cite{ElkinMatar} for constructing sparse near-additive spanners while here we use it for multiplicative spanners. It is a refinement of the degree sequence used in the algorithm of \cite{ElkinP01}.
	
	The second change in comparison to \cite{BarenboimEG15} is that here we use ruling sets from \cite{sew, KuhnMW18}, rather than those of \cite{AwerbuchGLP89}. This enables us to achieve constant $\kappa^{O({\log 1/\rho})}\cdot O\frho^{\frho+O(1)}$ (as long as $\kappa$ and $\rho$ are constant) stretch,  while still keeping the running time in check.
	
	Finally, algorithms of \cite{ElkinN17, ElkinMatar} as well as their precursor \cite{ElkinP01}, all construct spanners of size at least $\Omega(n{\log\log n})$, as each of the ${\log \log n}$ phases (or iterations) of these algorithms can potentially add $\Omega(n)$ edges. Here we show that in the case of \textit{multiplicative} spanners, as opposed to near-additive ones, constructed in \cite{ElkinP01, ElkinN17, ElkinMatar}, the total number of edges added on all phases combined can be bounded by $O(n)$. This is achieved by a careful accounting of edges added to the spanner.

	\subsection{Outline}
	Section \ref{sec preliminaries} provides necessary definitions for understanding this paper.
	 Our first spanners construction is described in Section \ref{sec polylog time spanner construction}. Our second spanners construction, which builds upon the first construction, is given in Section \ref{sec ultra-sparse}.
	
	\subsection{Preliminaries }\label{sec preliminaries}

	Given a graph $G= (V,E)$, a set of vertices $W\subseteq V$ and parameters $\alpha,\beta \geq 0$, a set of vertices $A\subseteq W$ is said to be
	a \textit{$(\alpha,\beta)$-ruling set} for $W$ if for every pair of vertices $u,v\in A$, the distance between them in $G$ is at least $\alpha$, and for every $u\in W$ there exists a representative $v\in A$ such that the distance between $u,v$ is at most $\beta$.

	Throughout this paper, we denote by $r_C$ the center of the cluster $C$ and say that $C$ is \textit{centered around} $r_C$. For a cluster $C$, define $Rad(C)= max\{d_H(r_C,v)\ | \ v\in C\}$, where $d_H$ is the distance matrix of the spanner $H$. For a set of clusters $P$, define 
	$Rad(P)=max\{Rad(C)\ | \ C\in P\}$.

\section{Polylogarithmic Time  Construction}\label{sec polylog time spanner construction}
	
	In this section, we devise a construction of 
	$O\left(	{\log n}\right)^{\kappa-1} $ spanner of size at most $\nfrac$. This construction requires 
	$ O\left({\log n}\right)^{\kappa-1}$ deterministic time in the \congestmo. For a constant $\kappa$, this yields a polylogarithmic stretch spanner, in time that is polylogarithmic in $n$.

	Section \ref{sec fa the const} contains a concise description of the algorithm. The technical details of the construction are discussed in Sections \ref{sec superclustering fast} and \ref{sec intercon fast}. Finally, the properties of the resulting spanner and the construction are analyzed in Section  \ref{sec fa analysis of const}.

	\subsection{The Construction}
	\label{sec fa the const}
	Our algorithm initializes $H$ as an empty spanner, and proceeds for $\ell+1$ phases. The parameter $\ell$ will be specified later. The input for each phase $i\in [0,\ell]$ is a collection of clusters $\mathcal{P}_i$. The input to phase $0$ is the partition of $V$ into singleton clusters, $\mpp{0} = \{ \{v\} \ |\ v\in V\}$.
	We say that a vertex $v\in V$ is \textit{popular} if it has neighbors from at least $\nko$ distinct clusters. A cluster $C\in \mathcal{P}_i$ is said to be popular if it contains a popular vertex.

	 Each phase $i$ consists of two steps. In the \textit{superclustering} step, popular clusters are merged into larger clusters. The new collection of large clusters will be the input for the next phase.
	 In the \textit{interconnection} step, clusters in $\mathcal{P}_i$ that have not been superclustered in this phase are interconnected to their neighboring clusters in $\mathcal{P}_i$.
	 
	 In the last phase $\ell$, we will ensure that the size of $\mathcal{P}_\ell$ is small enough, such that we can simply interconnect every pair of neighboring clusters, and therefore we will not form superclusters. Set  $\ell = \kappa-1$.

	\subsubsection{Superclustering}\label{sec superclustering fast}
	In this section we describe the superclustering step of phase $i$, for all $i\in [0,\ell-1]$. 
	
	Let $W_i$ be the set of popular clusters in $\mathcal{P}_i$. 
	First, for each cluster $C\in \mathcal{P}_i$, its cluster center $r_C$ broadcasts its ID to all vertices in $C$. 
	We then proceed to detect popular clusters. Each vertex $u\in C$ notifies all its neighbors that it belongs to the cluster centered around $r_C$. Each vertex $v$ that belongs to some cluster $C'\in \mathcal{P}_i$ now knows the IDs of all clusters it is adjacent to. If it has at least $\nko$ neighboring clusters (excluding $C'$ itself), it decides that it is popular, and informs its cluster center $r_{C'}$ that $C'$ is popular. Observe that after this communication terminates, every cluster center knows whether its cluster is popular or not, i.e., if the cluster belongs to $W_i$ or not. 
	
	We construct the virtual popular cluster graph $G'_i=(V',E')$, where the set of supervertices $V'=\mathcal{P}_i$ and $E'$ contains edges from each popular cluster to its neighboring clusters (whether they are popular or not). Define $\delta =2{\log n}$. We simulate the algorithm of Awerbuch at al. \cite{AwerbuchGLP89} on the graph $G'_i$ to construct $Q_i$, a $(3,\delta)$-ruling set for $W_i$. (Note that  $Q_i\subseteq W_i$.) See Section \ref{sec running time fast} below for the argument that this simulation can be conducted efficiently. Note that $Q_i$ is $3$-separated and $\delta$-ruling for the set $W_i$ with respect to the distance in $G'_i$.
	
	A BFS exploration is then simulated on $G'_i$ from all supervertices of $Q_i$ to depth $\delta$. As a result, a ruling forest $F'_i$ is constructed. Each supervertex $C'\in V'$ that is spanned by a tree in $F'_i$ that originated in a cluster $C\in Q_i$, now becomes superclustered into the cluster $C$. The cluster $C'$ now 
	decides which cluster $C'_{pred}$ is its predecessor with respect to $F'_i$. 
	Observe that the cluster $C'_{pred}$ discovered $C'$ because there is an edge $(C'_{pred},C')$ that was traversed in the BFS exploration, and it was the first edge to reach $C'$. This edge was simulated by an original edge $(u',u)\in E$ such that $u\in C'$ and $u'\in C'_{pred}$. The cluster $C'$ chooses to add the edge $(u',u)$ to the spanner $H$.
	 This concludes the description of the superclustering step (see Figure \ref{fig superclustering charge} for an illustration). 
	
	\begin{figure}
		\begin{center}
			\includegraphics[scale=0.2]{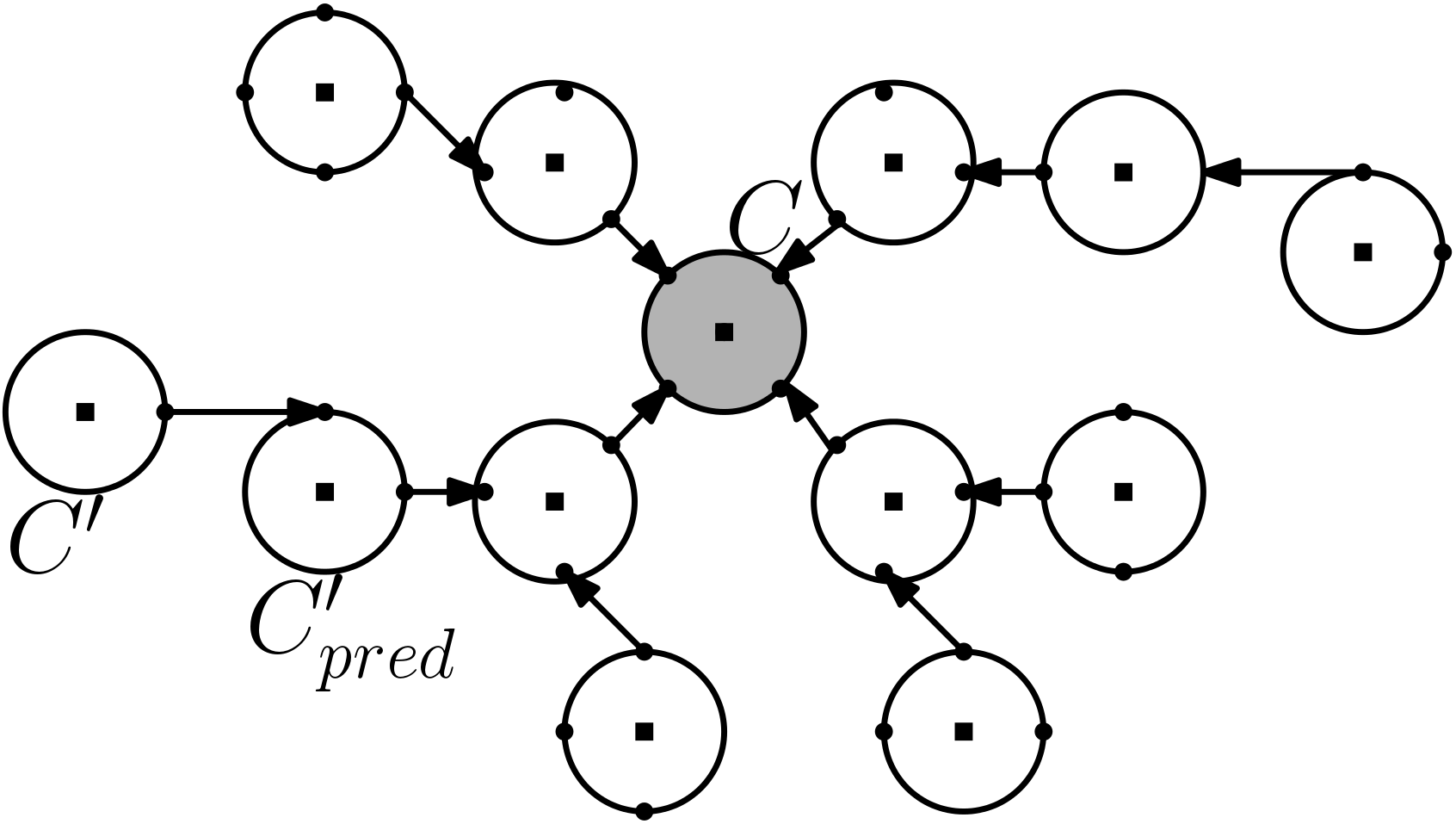}
			\caption{The superclustering edges. The gray circle represents a cluster $C\in Q_i$. The white circles represent the clusters $C'\in \mpi\setminus Q_i$ that are clustered into the supercluster around $C$. The edges represent the superclustering edges. The arrow from a cluster $C'$ to $C'_{pred}$ represents that $C'$ is charged for the edge between $C',C'_{pred}$. }
			\label{fig superclustering charge}
		\end{center}	
	\end{figure}
	
	We will show that all popular clusters w.r.t. to $\mathcal{P}_i$ are superclustered in phase $i$. 
	Recall that every center of a popular cluster knows that it is popular, before the construction of the ruling set $Q_i$.
	
	\begin{lemma}
		\label{lemma popular are clustered fst}
		For every phase $i\in [0,\ell-1]$, all popular clusters in $\mathcal{P}_i$ are superclustered into clusters of $\mathcal{P}_{i+1}. $
	\end{lemma} 

	\begin{proof}
		Let $C$ be a popular cluster. The set $Q_i$ is a $(3,\delta)$-ruling set for $W_i$ in the virtual graph $G'_i$. Therefore, there exists a supervertex $C'\in Q_i$ with distance at most $\delta$ from $C$ in the virtual graph $G'_i$. Then, the BFS exploration that originated from all supervertices of $Q_i$ to depth $\delta$ in $G'_i$ discovered $C$, and so it is superclustered into a cluster of $\mathcal{P}_{i+1}$.
	\end{proof}

	We will now bound the radius of the clusters collection $\mathcal{P}_i$. Define recursively 
	\begin{equation}
	\label{eq define ri gen}
	R_0=0 \ 	\textit{ and } \ R_{i+1} = (2\delta+1)R_i+\delta.
	\end{equation}

	\begin{lemma}\label{lemma single tree}
		Let $i\in [0,\ell]$ and let $C$ be a cluster of $\mathcal{P}_i$. 
		At the beginning of phase $i$, the spanner $H$ contains a spanning tree $T_{C}$
		such that for every vertex $u\in V(C)$, there is a path in $T_{{C}}$ from $u$ to the cluster center $r_{C}$, of length at most $R_{i}$, that contains only vertices from $V({C})$. 
	\end{lemma}

\begin{proof}
	The proof is by induction on the index of the phase $i$.
	For $i=0$, all clusters in $\mathcal{P}_0$ are singletons, and so the claim is trivial. 
	We will assume that the claim holds for phase $i\geq 0$, and prove it holds for $i+1$.
	
	Let $\widehat{C}$ be a cluster in $\mathcal{P}_{i+1}$. Observe that the cluster $\widehat{C}$ contains a cluster $C\in Q_i$, and some other clusters $C'\in \mathcal{P}_i$, that have been reached by the BFS exploration in $G'_i$ that originated from the supervertex $C$. 
	Denote $T'_{C}=(V(\widehat{C}),E'_{\widehat{C}})$ the tree constructed by the BFS exploration that originated in the supervertex $C$. 
	Let $u\in V$ be a vertex in $V(\widehat{C})$.
	
	\textbf{Case 1:} If $u\in V(C)$, then by the induction hypothesis the spanner $H$ contains a spanning tree $T_C$ such that there is a path in $T_C$ from $u$ to the cluster center $r_C$, of length at most 
	$R_i$ that contains only vertices from $V(C)$. Since $T_C$ is in $T_{\widehat{C}}$, there is a path in $T_{\widehat{C}}$ from $u$ to the cluster center $r_C$, of length at most $R_i$ that contains only vertices from $V(C)$. Since $\delta>0$, $R_i\leq R_{i+1}$, and the claim holds. 
	
	\textbf{Case 2:} If $u\notin V(C)$, then $u$ belongs to some cluster $C'\in \mathcal{P}_i$, $C'\neq C$, that was superclustered into $\widehat{C}$. 
	Let $P'= (C=C_0,C_1,\dots, C_q=C')$ be the path in $T'_C$ from $C$ to $C'$. Denote $v_0,v_1,\dots,v_q$ the respective centers of the clusters in $P'$. 
	For each pair $C_{j-1},C_{j}$ for $j\in [q]$, we know that $C_{j-1}$ is the predecessor of $C_{j}$ w.r.t. to $F'_i$. Then, the center $v_{j}$ added to $H$ an edge $(y_{j-1},x_{j})$ such that $y_{j-1}\in C_{j-1}$ and $x_{j}\in C_{j}$. 
	Moreover, by the induction hypothesis, for every $j\in [0,q]$, there is a path in $H$ from the center $v_j$ to all vertices in $C_j$ of length at most $R_i$, that contains only vertices from $C_j$. 
	
	Let $P$ be the path in $H$ obtained by replacing the edges in $P'$ with edges from $H$ as follows. Define $v_0=x_0$ and $u=y_q$. The path $P$ begins with the path of length at most $R_i$ from $r_C$ to $y_0$ in the spanner $H$. Then, for each $j\in [q]$, the edge $(C_{j-1},C_{j})$ is replaced with the edge $(y_{j-1},x_{j})$ and the path in $H$, of length at most $2R_i$, from $x_j$ to $y_j$. See Figure \ref{fig path p} for an illustration. Observe that for all $j\in [0,q]$, the vertices of $C_j$ also belong to $\widehat{C}$. Thus the path $P$ is a path from $v_0$ to $u$, of length at most $R_i+q+2q\cdot R_i$, that contains only edges from $\widehat{C}$. Since $q\leq \delta$, we obtain 	
	$$ d_H(r_C,u) \leq R_i + 2R_i \cdot \delta + \delta = (2\delta+1)R_i+\delta = R_{i+1}. $$
\end{proof}

\begin{figure}
	\begin{center}
		\includegraphics[scale=0.3 ]{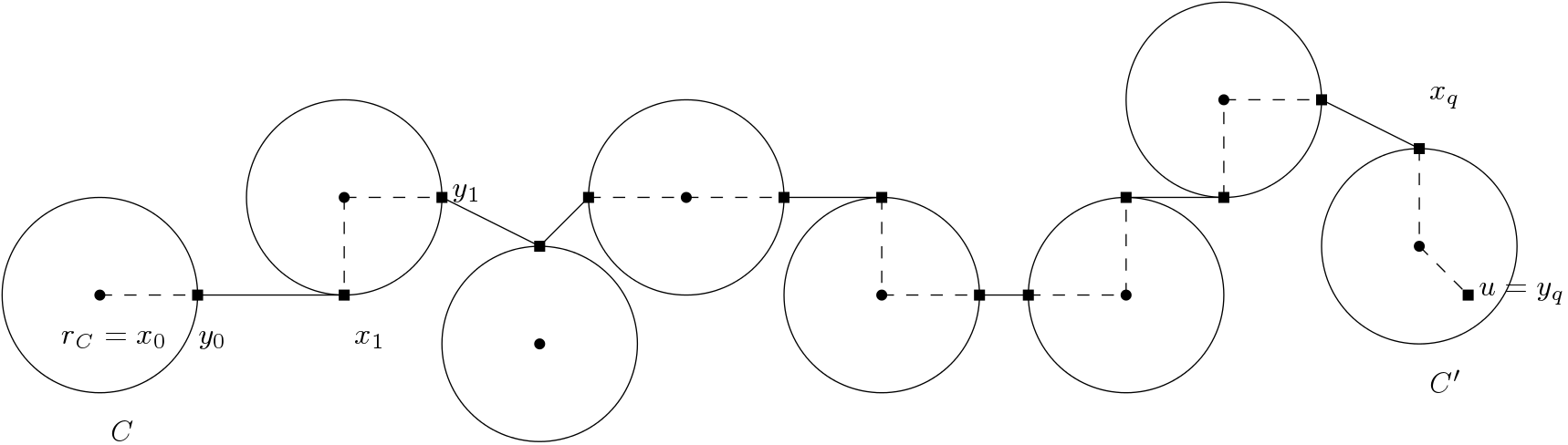} 
	\end{center}
	\caption{The path between a vertex $u\in C' $, such that $C'$ has been superclustered into $\widehat{C}$, to the cluster center $r_C = r_{\widehat{C}}$. In the figure, the circles represent clusters on the path $P'$ from $C$ to $C'$ and the round vertices represent their respective centers. The square vertices represent vertices in the clusters that are not necessarily cluster centers. The solid line represent edges between clusters. The dashed lines represent paths within clusters. }
	\label{fig path p}
\end{figure}
Observe that Lemma \ref{lemma single tree} implies that for all $i\in [0,\ell]$, we have:

\begin{equation}
\label{eq ri bouns pi gen}
Rad(\mathcal{P}_i)\leq R_i
\end{equation}

\subsubsection{Interconnection}\label{sec intercon fast}
	
We now discuss the details of the execution of the interconnection step of a phase $i\in [ 0, \ell]$.

Denote by $U_i$ the set of clusters of $\mathcal{P}_i$ which were not superclustered into clusters of $\widehat{\mpi}$. For phase $\ell$, the superclustering step is skipped. Therefore, we set $U_\ell = \mathcal{P}_\ell$. 

Recall that a vertex $v\in V$ and a cluster $C$ are said to be neighbors if there exists an edge $(v,u)\in E$ such that $u\in C$. 
Let $v\in C$ be a vertex such that $C\in U_i$. In the interconnection step of phase $i$, the vertex $v$ will add to the spanner $H$ an edge to each one of its neighboring clusters. 
If $v$ has multiple neighboring vertices that belong to the same cluster, it will arbitrarily choose one of them to add an edge to. This concludes the description of the interconnection step of phase $i$ (see Figure \ref{fig inter charge} for an illustration).

\begin{figure}
	\begin{center}
		\includegraphics[scale=0.2]{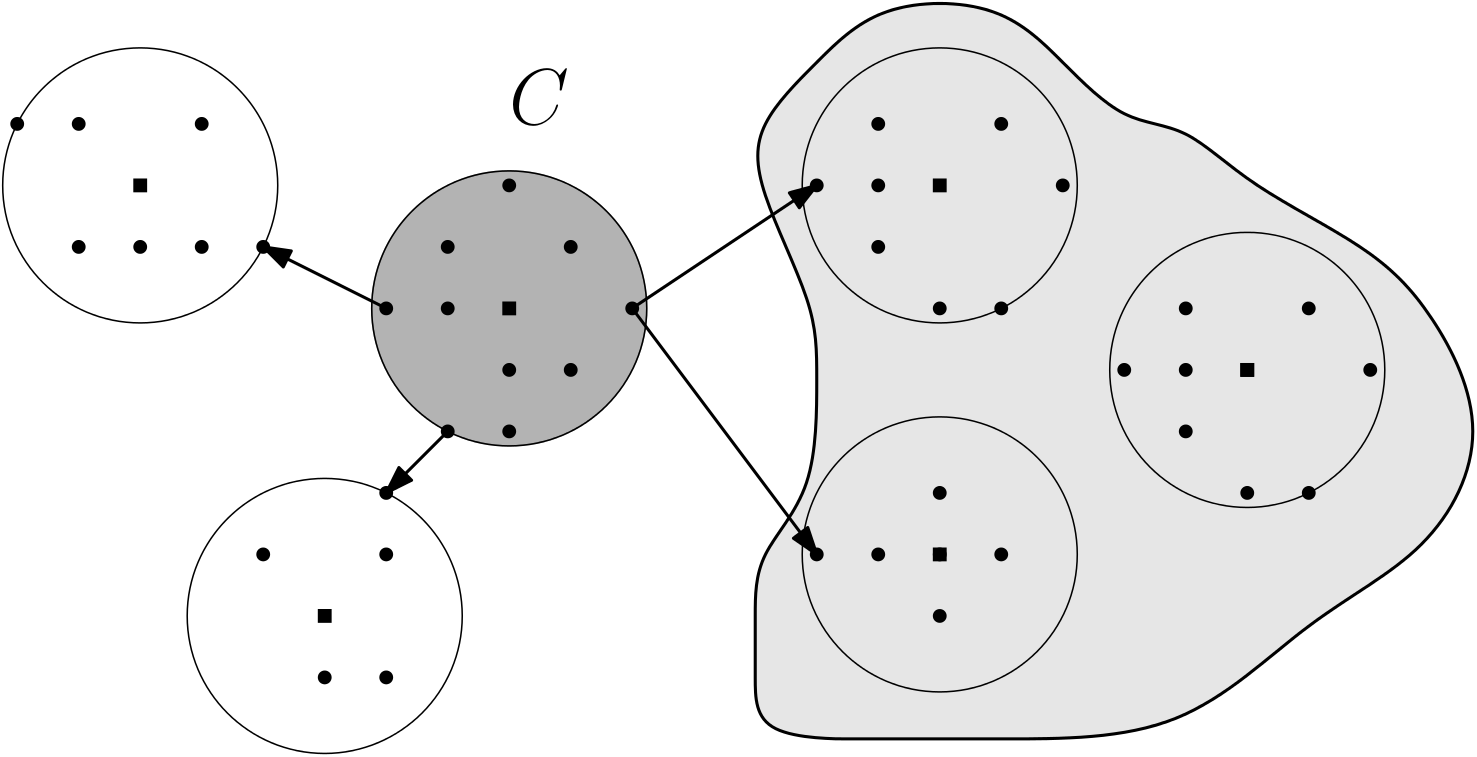}
		\caption{The interconnection edges. The dark gray circle represents a cluster $C\in U_i$. The other circles represent clusters $C'\in \mpi$. The light gray area represent a supercluster that was formed in this phase. The arrowed segments represent the interconnection edges that the center $r_C$ of the cluster $C$ is charged for. }
		\label{fig inter charge}
	\end{center}	
\end{figure}

Denote by $U^{(i)}$ the union of all sets $U_0,U_1,\dots,U_i$, i.e., $U^{(i)} = \bigcup_{j=0}^iU_j$. Observe that the set $U^{(\ell)}$ is a partition of $V$.
 
\subsection{Analysis of the Construction}\label{sec fa analysis of const}
In this section, we provide an analysis of the parameters of the resulting spanner and the running time of the algorithm. The following lemmas provide an explicit bound on $R_\ell$, that will be used in the following sections.

Recall that $R_0=0$ and $R_{i+1}= (2\delta+1)R_i+\delta$ (see \cref{eq define ri gen}). 

\begin{lemma}
		 \label{eq explicit ri gen}
	For every $i\in [0,\ell]$, we have $$R_i = \delta\cdot\sum^{i-1}_{j=0} (2\delta+1)^{j} .$$
\end{lemma}

\begin{proof}
	The claim is proved by induction on the index of the phase $i$. For $i=0$,	the claim 
	 is trivial as both sides of the equation are equal to $0$.
	 
	 Assume that the claim holds for $i\geq 0$. We will prove the claim for $i+1$. 
	 By definition and the induction hypothesis, we have 
	 
	 \begin{equation*}
	 \begin{array}{cllll}
	 R_{i+1} &=& (2\delta+1)R_i+\delta\\
	 &=& (2\delta+1)\left(\delta\cdot\sum^{i-1}_{j=0} (2\delta+1)^j \right)+\delta\\
	 &=& \left(\delta\cdot\sum^{i-1}_{j=0} (2\delta+1)^{j+1} \right)+\delta\\
	 	 &=& \delta\cdot\sum^{i}_{j=0} (2\delta+1)^{j} \\
	 \end{array}
	 \end{equation*}
	 
\end{proof}

Observe that Lemma \ref{eq explicit ri gen}, implies that $R_i\leq \frac{1}{2}\cdot(2\delta+1)^{i}$ since:

\begin{equation}\label{eq explicit Ri gen}
\begin{array}{cllll}
R_i &=&\delta\cdot\sum^{i-1}_{j=0} (2\delta+1)^{j}\\

&=&\delta\cdot\left[
\frac{(2\delta+1)^{i}-1}{(2\delta+1)-1}
\right]\\
&\leq &\frac{1}{2}\cdot (2\delta+1)^{i}

\end{array}
\end{equation}

Recall that $\delta = 2{\log n}$. Then,

\begin{equation}\label{eq explicit Ri fast}
\begin{array}{cllll}
R_i &\leq &\frac{1}{2}\cdot (4{\log n}+1)^{i}

\end{array}
\end{equation}

\subsubsection{Analysis of the Stretch}

In this section, we analyze the stretch of the resulting spanner. Consider an edge $(u,v)\in E$. Since $U^{(\ell)}$ is a partition of $V$, both $u$ and $v$ belong to clusters of $U^{(\ell)}$. The following two lemmas bound the stretch of the edge $(u,v)$ in the spanner, in the case where $u,v$ belong to the same cluster in $U^{(\ell)}$, and in the case where $u,v$ belong to different clusters of $U^{(\ell)}$, respectively.

\begin{lemma}\label{lemma stretch fa1}
	Let $(u,v)$ be an edge in the original graph $G$, such that $u,v$  belong to the same clusters of $U^{(\ell)}$. Then,
	$$	d_H(u,v) \leq 2R_\ell. $$
\end{lemma}	

\begin{figure}
	\begin{center}

		\includegraphics[scale=0.1]{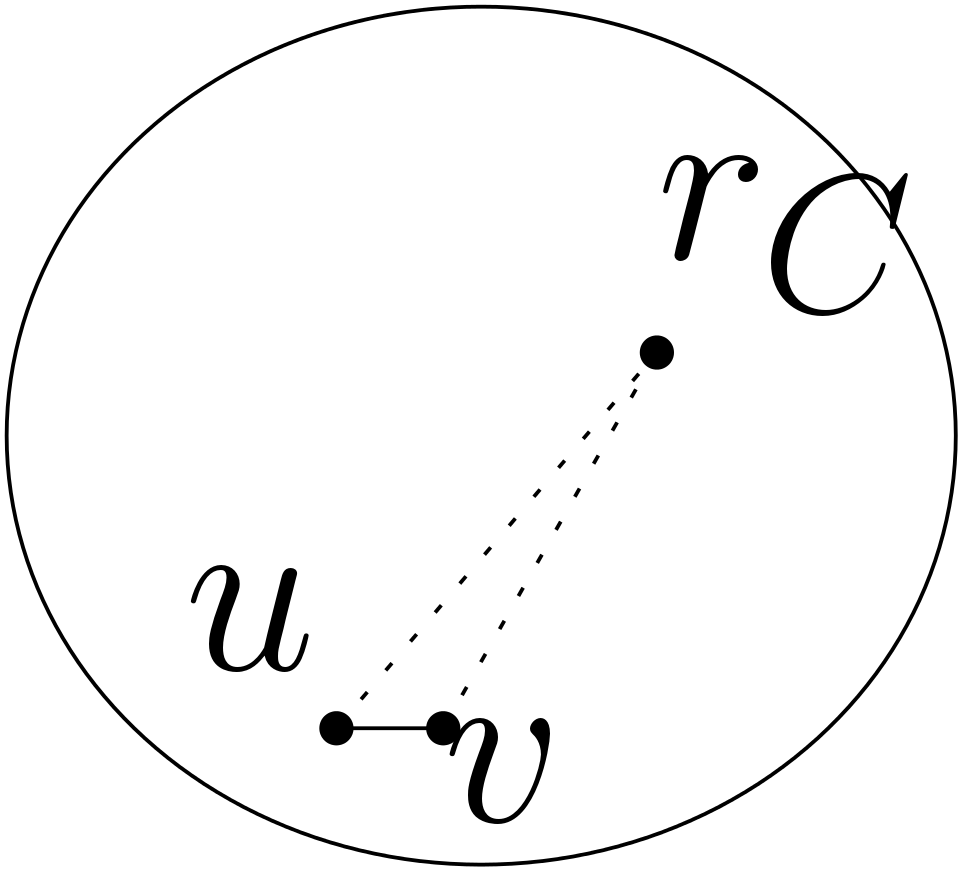}
		\caption{The path from $u$ to $v$ in $H$, if both $u,v$  belong to the cluster $C$, centered around $r_C$. The solid line represents the edge in $G$ between $u,v$. The dashed line represents the paths in $H$ from the cluster center $r_C$ to the vertices $u,v$.}
		\label{fig same_cluster}	
\end{center}
	
\end{figure}

\begin{proof}
	 The vertices $u,v$ both belong to the same cluster $C$ in phase $i$. Thus, there are paths in $H$ from $u$ and $v$ to the center $r_C$ of length $R_i$ (see Figure \ref{fig same_cluster} for an illustration). It follows that there is a path in $H$ between $u,v$ of length at most $2R_i$ in $H$. Since $i\leq \ell$, we have 
		$$	d_H(u,v) \leq 2R_\ell. $$
\end{proof}

\begin{lemma}\label{lemma stretch fa2}
	Let $(u,v)$ be an edge in the original graph $G$, such that $u,v$  belong to different clusters of $U^{(\ell)}$. Then,
	$$	d_H(u,v) \leq 2R_\ell+1. $$	
\end{lemma}

\begin{figure}
	\begin{center}
		\includegraphics[scale=0.15]{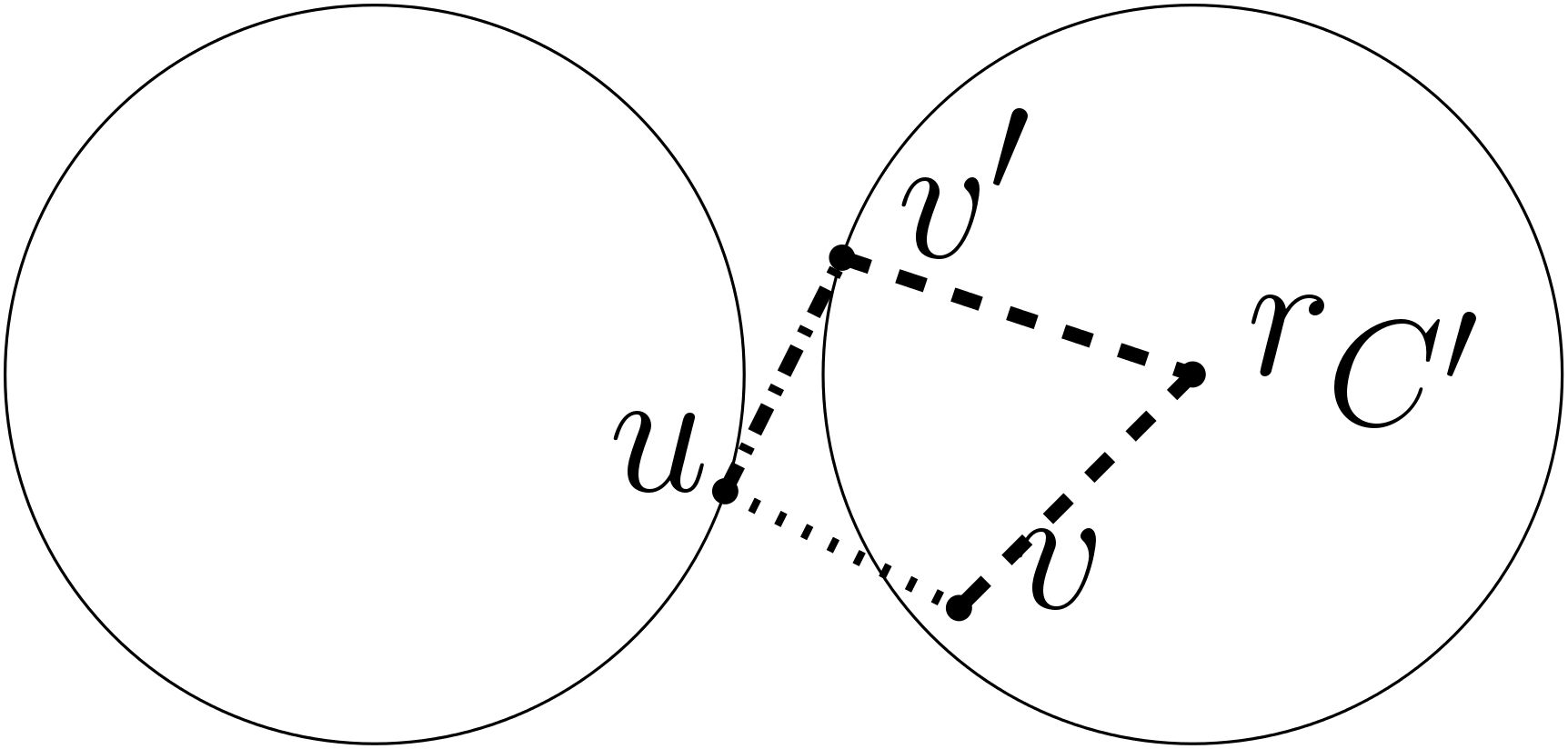}
		\caption{The path from $u$ to $v$ in $H$, if $u\in C$ and $v\in C'$, with $C\in U_i$ and $C'\in \mathcal{P}_i$. The dotted line between $u,v$ represents the original $(u,v)$ edge fro $G$. The dashed-dotted line represents the edge between $(u,v')$ that belongs to $G$ and to $H$. The dashed lines represent paths in $H$ from the cluster center $r_{C'}$ to the vertices $v,v'$. }
		\label{fig not same_cluster_fast}
	\end{center}
	
\end{figure}

\begin{proof}
	Let $C$ be the cluster such that $C\in U_i$ and $u\in C$. 
	Let $C''$ be the cluster such that $C''\in U_j$ and $v\in C''$. Assume w.l.o.g. that $i\leq j$. 	
	Let $C'$ be the cluster such that $C'\in \mathcal{P}_i$ and $v\in C'$. (Observe that if $i=j$, then $C''=C'$.) 
	
	Since  $C\in U_i$, in the interconnection step of phase $i$, the vertex $u$ added edges to all its neighboring clusters. Specifically, it added an edge to a vertex $v'\in C'$. 
	There are paths in $H$ from $v,v'$ to their cluster center $r_{C'}$ of length at most $R_i$. It follows that there is a path in $H$ between $u,v$ of length at most $2R_i+1$. Since $i\leq \ell$, we have 
	$$	d_H(u,v) \leq 2R_\ell+1. $$
\end{proof}

As a corollary to Lemmas \ref{lemma stretch fa1} and \ref{lemma stretch fa2}, we have:

\begin{corollary}\label{lemma stretch fa}
	For every edge $(u,v)\in E$ it holds that 
	\begin{equation}
		\label{eq stretch fast2}
		d_H(u,v) \leq 2R_\ell+1
	\end{equation}
\end{corollary}

We will now derive an explicit expression of the stretch. 
Recall that 
$\ell = \kappa-1$.
By \cref{eq explicit Ri fast,eq stretch fast2}, it follows that for every edge $(u,v)\in E$, we have

\begin{equation}
\begin{array}{clllll}

d_H(u,v)&\leq& 2R_\ell+1\\
&\leq& 2\left( 1/2\cdot \left(4{\log n}+1\right)^{\kappa-1}\right)+1\\
&\leq& \left(4{\log n}+1\right)^{\kappa-1}+1.

\end{array}
\end{equation}

Therefore, for every pair of vertices $x,y\in V$, the distance between $x,y$ in $H$ satisfies:

\begin{equation}\label{eq stretch fa}
d_H(x,y)\leq
\left(	 \left(4{\log n}+1\right)^{\kappa-1}+1\right)
\cdot d_G(x,y).		
\end{equation}

\subsubsection{Analysis of the Number of Edges}\label{sec fa analysis size}

In this section, we analyze the size of the spanner $H$. We will charge each edge in the spanner $H$ to a single vertex, and show that over all phases of the algorithm, a vertex is charged for at most $\nfrac$ edges. 

Observe that $H$ contains two types of edges, the \textit{superclustering edges}, and the \textit{interconnection edges}.

A superclustering edge that is added in a phase $i$ is an edge that connects a cluster $C_j\in \mathcal{P}_i\setminus Q_i$ to its predecessor in the BFS forest $F_i$. 
We will charge this edge on the center $r_{C_j}$ of the cluster $C_j$. For example, if for some $h\geq 1$, clusters $C_1,C_2,\dots,C_h$, centered at vertices $v_1,v_2,\dots,v_h$, respectively, are clustered into a supercluster rooted at a cluster $C$, then each of these centers $v_1,v_2,\dots,v_h$ is charged for a single edge. Note that the center $r_C$ of the cluster $C$ is not charged for any edges. Furthermore, since all clusters $C_1,C_2,\dots,C_h$ have been superclustered into the new supercluster centered around $C$, each cluster center $v_1,v_2,\dots,v_h$ will not be a cluster center in future, and will not be charged in this way ever again.
 See Figure \ref{fig superclustering charge} for an illustration. We summarize this argument in the following lemma.

\begin{lemma}
	\label{lemma fast one charge sup}
	Each vertex $v\in V$ is charged for at most one superclustering edge.
\end{lemma}

An interconnection edge added in phase $i$ is an edge added by a vertex that belongs to a cluster $C\in U_i$. We will charge each interconnection edge to the vertex $v$ which added it to the spanner $H$ (see Figure \ref{fig fast intercon} for an illustration). Recall that $U^{(\ell)}$ is a partition of $V$. Hence for each vertex $v\in V$, there is exactly one phase $i\in [0,\ell]$ such that $v$ belongs to a cluster $C\in U_i$. We begin by analyzing the number of edges added to the spanner $H$ by all phases other than the concluding phase $\ell$.

\begin{figure}
	\begin{center}

		\includegraphics[scale=0.25]{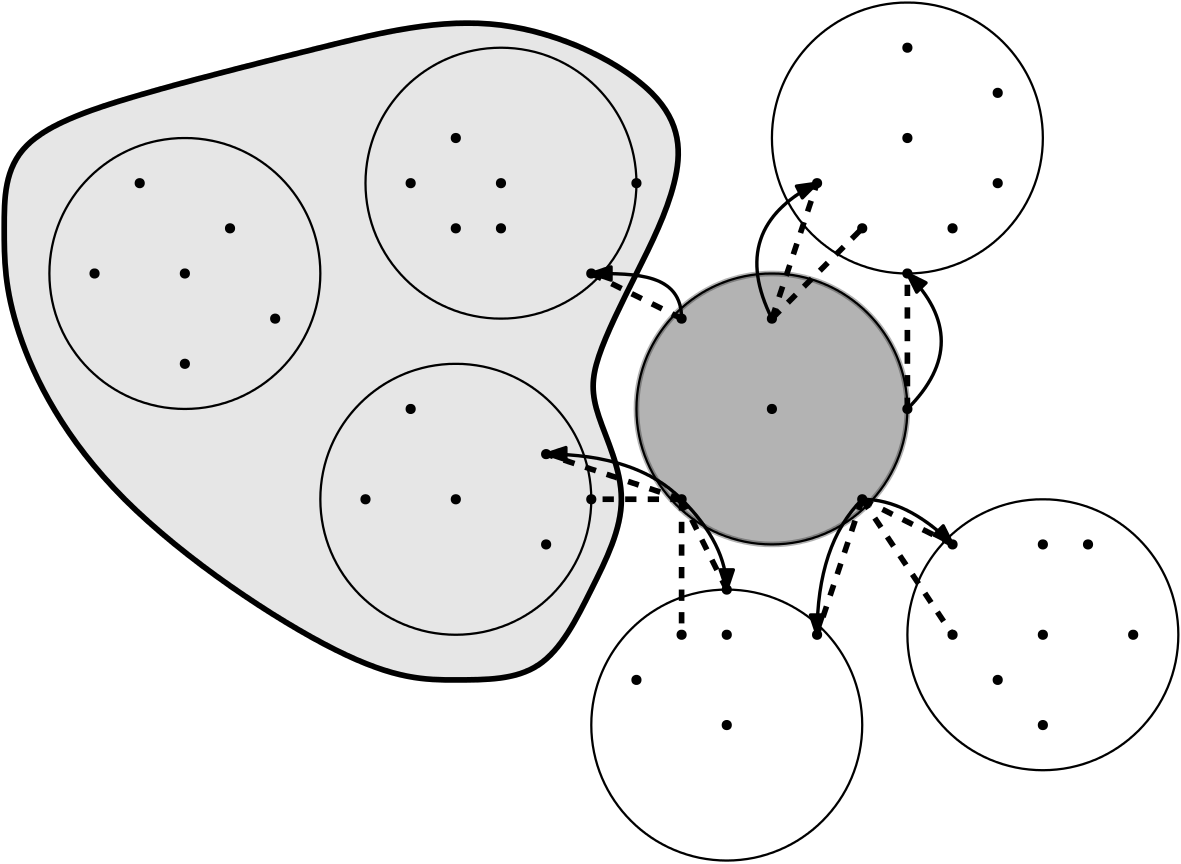}
		\caption{The interconnection edges. The dark gray circle represents a cluster $C\in U_i$. The other circles represent clusters $C'\in \mpi$. The light gray area represent a supercluster that was formed in this phase. The dotted lines represent edges in the graph $G$. The arrows represent the interconnection edges. Each vertex in the cluster $C$ is charged for all the interconnection edges that are incident to the vertex. }
		\label{fig fast intercon}	
\end{center}
\end{figure}

\begin{lemma}
	\label{lemma fast one charge int}
	Let $v$ be a vertex that belongs to a cluster $C\in U^{(\ell-1)}$. Then, the vertex $v\in V$ is charged for less than $\nko$ interconnection edges.
\end{lemma}

\begin{proof}
	Let $i\in [0,\ell-1]$ be a phase such that $v$ belongs to a cluster $C\in U_i$. By Lemma \ref{lemma popular are clustered fst}, we have that $C$ is not popular. Then, by definition, $v$ has neighbors from less than $\nko$ other clusters. Thus in the interconnection step of phase $i$, the vertex $v$ adds less than $\nko$ edges to the spanner $H$ in phase $i$.

	Since there is only one single phase $i$ such that $v$ belongs to a cluster of $U_i$, we have that $v$ is charged for less than $\nko$ interconnection edges throughout all phases of the algorithm. 
\end{proof}

Recall that in phase $\ell$, we set $U_\ell \gets \mathcal{P}_\ell$. Each vertex $v$ that belongs to a cluster $C\in U_\ell$ now adds to $H$ interconnection edges from it to all of its neighboring clusters. It follows that each vertex $v$ that belongs to a cluster $C\in U_\ell$ is charged for at most $|U_\ell|-1$ edges.

We now provide an upper bound on the size of $\mathcal{P}_\ell= U_\ell$.

\begin{lemma}
	\label{lemma size pl fast}
	For all $i\in [0,\ell]$, the size of $\mathcal{P}_i$ is at most $n^{\frac{\kappa-i}{\kappa}}$.
\end{lemma}

\begin{proof}
	The proof is by induction on the index of the phase $i$.
	For $i=0$, we have $|\mathcal{P}_0|=n$, and so the claim is trivial. 
	
	We will assume that the claim holds for phase $i<\ell$, and prove it holds for $i+1$.
	Let $\widehat{C}$ be a cluster in $\mathcal{P}_{i+1}$.	
	We know that it was constructed in phase $i$, by the BFS exploration that originated in a cluster $C$. Define $\Gamma_{\mpi}(C)$ the set of all neighboring clusters of $C$ from $\mpi$. Define $\widehat{\Gamma}_{\mpi}(C) = \{C\}\cup \Gamma_{\mpi}(C)$ i.e., the set of neighbors of $C$ as well as $C$ itself. Since $Q_i\subseteq W_i$, we have that $C$ is popular, i.e., by definition, there exists a vertex $v\in C$ that has neighbors from at least $\nko$ different clusters. Hence, $|\widehat{\Gamma}_{\mpi}(C)|\geq \nk$. Moreover, since $Q_i$ is a $3$-separated in the graph $G'_i$, for every pair of distinct clusters $C,C'\in Q_i$, the sets $\widehat{\Gamma}_{\mpi}(C),\widehat{\Gamma}_{\mpi}(C')$ are disjoint. 
	Thus, by the induction hypothesis we have:

	\begin{equation}
	\label{eq bound size}
	|\mathcal{P}_{i+1}|\leq \frac{|\mathcal{P}_i|}{\nk} \leq \frac{n^{\frac{\kappa-i}{\kappa}}}{\nk} = n^{\frac{\kappa-(i+1)}{\kappa}} .
	\end{equation}
	
\end{proof}

By Lemma \ref{lemma size pl fast} we have $|\mathcal{P}_{\kappa-1}| \leq \nk$. Therefore, each vertex that adds edges in the interconnection step of the concluding phase $\kappa-1$, adds at most $\nk-1 $ edges.

It follows that each vertex $v\in V$ is charged for at most one superclustering edge, and at most $\nko$ interconnection edges by all phases of the algorithm \textit{combined}. Thus, the size of the spanner $H$ is bounded by: 
\begin{equation} 
\label{eq size fa}
|H| \leq \nfrac.
\end{equation}

\subsubsection{Analysis of the Running Time}\label{sec running time fast}

We begin by analyzing the running time of a single phase $i\in [0,\ell-1]$. 

\textbf{Superclustering.} The superclustering step of phase $i$ begins with detecting the popular clusters. Each center $r_C$ of a cluster $C$ broadcasts its ID to the entire cluster $C$. By Lemma \ref{lemma single tree} this requires at most $R_i$ time. Then, each vertex acquires the information regarding its neighboring clusters in one communication round. It now locally decides whether it is popular or not. This requires zero time. If a vertex decided that it is popular, it will send this message to its cluster center. 
This involves a convergecast in the spanning tree of the cluster, and by Lemma \ref{lemma single tree}, this requires at most $R_i$ time. Therefore, the detection of popular clusters requires $O(R_i)$ time. 

 The algorithm of Awerbuch et al. \cite{AwerbuchGLP89} computes a $(3,2{\log n})$-ruling set in a graph $G$ on $n$ vertices in $O({\log n})$ time. Since in the algorithm of \cite{AwerbuchGLP89} every vertex sends the same message on each round to all its neighbors, the algorithm applies to the Broadcast-CONGEST model. (See Appendix \ref{apend rs} for the details of implementation). Therefore, the algorithm can be simulated on a supergraph $G'_i$, where the overhead is the maximum diameter of a supervertex in a simulated supergraph. Therefore computing a $(3,2{\log n})$ requires $O(R_i\cdot {\log n})$ time. 
 
 The BFS exploration to depth $2{\log n}$ in $G'_i$ requires $O(R_i\cdot {\log n})$ time. 
 
 \textbf{Interconnection.} The interconnection requires only a single round, as each vertex knows whether it belongs to a cluster of $U_i$ or not. Therefore, each vertex that adds edges to its neighbors only needs to inform them that it has added this edge. 

 Therefore, the running time of a single phase  $i\in [0,\ell-1]$ is $O(R_i\cdot {\log n})$. For the concluding phase $\ell$, note that we do not form superclusters. At the beginning of the phase, each cluster center broadcasts its ID to all vertices in its cluster. This requires at most $R_\ell$ time. Then, each vertex sends a single message to its neighbors, and  decides locally which  edges to add to the spanner $H$. Thus the running time of the concluding phase is $O(R_\ell)$.

 It follows that the running time of the entire algorithm is 
 \begin{equation}
 \label{eq alg fat running time}
 \begin{array}{cllll}
 O\left(R_\ell+
 \sum^{\kappa-2}_{j=0} R_j\cdot {\log n}
 \right) &=&

 O\left(R_\ell+ 
 (1/2){\log n}\sum^{\kappa-2}_{j=0} (4{\log n}+1)^j
 \right) \\&=&
 O\left(\frac{1}{2}\cdot (4{\log n}+1)^{\kappa-1}+
 {\log n}\left[\frac{(4{\log n}+1)^{\kappa-1}-1}{(4{\log n}+1)-1}\right] 
 \right)\\&=&
 {O}\left({\log n} \right)^{\kappa-1}.
 \end{array}
 \end{equation}

As a corollary to \cref{eq stretch fa,eq size fa,eq alg fat running time}, we conclude:

\begin{corollary}
	\label{corollary polylog construction}
	For any parameter  $\kappa\geq 2$,  and any $n$-vertex graph $G=(V,E)$, our algorithm constructs a $\left((4{\log n}+1)^{\kappa-1}+1\right)$-spanner with at most 
	$ \nfrac $
	edges,
	in 
	$ {O}({\log n})^{\kappa-1} $
	deterministic time in the CONGEST model.
	
\end{corollary}
	
%
%
%
%
\section{A Construction of Sparse Spanners and Linear-Size Skeletons}\label{sec ultra-sparse}
 
	Hence we aim at a low polynomial time, i.e., $O(n^\rho)$ for an arbitrarily small constant $\rho>0$. This increased running time enables us to modify the algorithm described in Section \ref{sec polylog time spanner construction}, to obtain a spanner with much better parameters. Specifically, 
	we will show that 
	for any parameters $\kappa\geq 2$, and $\frac{1}{\kappa}\leq \rho<\frac{1}{2}$, and any $n$-vertex unweighted undirected graph $G=(V,E)$, our algorithm constructs a $t $-spanner with
	$ \nfrac $	edges,
	in 
	$ O\left(n^\rho\cdot t \right)$
	deterministic time in the CONGEST model, where $t =
	\left(\frac{4}{\rho}+1\right)^{ {\log \kappa\rho}+\frac{1}{\rho}+O(1)}
	$. 
	
	The stretch $t$ can be written as
	$O\left(\frac{1}{\rho}\right)^{\left(\frac{1}{\rho}\right)+O(1)}
	\cdot O\left(
	(\kappa\rho)^{\log (4/\rho+1)}
	\right)$, i.e., polynomial in $\kappa$ as long as $\rho= \Omega(1)$ is an arbitrarily small constant.
	In particular, by setting $\kappa = \omega({\log n})$, we obtain a $polylog(n)$-spanner of size $O(n)$ in deterministic \congest\ time $O(n^{\rho})$, for any arbitrarily small constant $\rho >0$.

	Section \ref{sec st the const} contains a concise description of the algorithm. The technical details of the construction are discussed in Section \ref{sec st the const}. Finally, the properties of the resulting spanner and the construction are analyzed in Section  \ref{sec st Analysis of the Construction}.
\subsection{The Construction}\label{sec st the const}

Our algorithm initializes $H$ as an empty spanner, and proceeds for $\ell+1$ phases. The parameter $\ell$ will be specified in the sequel. The input for each phase $i\in [0,\ell]$ is a collection of clusters $\mathcal{P}_i$ and a degree threshold parameter $deg_i$. The input for phase $0$ is the partition of $V$ into singleton clusters.

As in the construction from Section \ref{sec polylog time spanner construction} of this paper, each phase of the current construction also consists of a superclustering step and an interconnection step. In the superclustering step of each phase, clusters that have many neighboring clusters are merged into superclusters. In the interconnection step, clusters of low degree are interconnected to all their neighboring clusters.

Denote by $\Gamma_{\mathcal{P}_i}(C)$ the set of clusters $C'\in \mathcal{P}_i$ such that $C,C'$ are neighboring clusters. 
A cluster $C$ is said to be \emph{popular} if it has at least $deg_i$ neighboring clusters, i.e., if $|\Gamma_{\mathcal{P}_i}(C)| \geq deg_i$. Observe that this definition differs from the definition of popular clusters in Section \ref{sec polylog time spanner construction}, where we could not deliver multiple messages from a vertex to its cluster center. In the previous construction, each vertex had to decide for itself whether it is popular or not, and inform the cluster center of its decision. In the current algorithm, each cluster center can aggregate the information that resides within the vertices of its cluster, and make a decision.

Similarly to the construction described in Section \ref{sec polylog time spanner construction}, we aim to have the size of $\mathcal{P}_\ell$ at most $n^\rho$. (One can think of $\rho$ in the previous construction as equal to $1/\kappa$.) This will ensure that in the concluding phase $\ell$, even if every pair of clusters in $\mathcal{P}_\ell$ are interconnected by an edge, we will still not add too many edges to the spanner. Therefore, we will not construct superclusters in phase $\ell$.

Set the maximum index of a phase $\ell$ by 
$\ell= \floor*{{\log \kappa\rho}}+ \lceil\frac{\kappa+1}{\kappa\rho}\rceil-1$, as in \cite{ElkinN17,ElkinMatar}. (A similar approach was employed there for constructing \textit{near-additive} spanners. Here we employ it for building sparse multiplicative spanners.)
The execution of each phase $i$ of our algorithm requires at least $deg_i$ time, and we aim at running time of at most $n^\rho$. Therefore, we partition phases $0,1,\dots,\ell-1$ into two stages, the \emph{exponential growth stage} and the \emph{fixed growth stage}.
In the \textit{exponential growth stage}, that consists of phases 
$ 0,\ldots,i_0=\floor*{{\log (\kappa\rho)}}$, we set $deg_i = n^\frac{2^i}{\kappa}$. 
In the \textit{fixed growth stage}, which consists of phases $ i_0+1,\ldots,i_1 = i_0+\lceil\frac{\kappa+1}{\kappa\rho}\rceil-2=\ell-1$, we set $deg_i = n^\rho$. Observe that for every index $i$, we have $deg_i\leq n^\rho$.

The concluding phase $\ell$ is not a part of either of these two stages. We will show that the size of $P_\ell$ is small enough, such that we do not need to form superclusters in phase $\ell$. In phase $\ell$, we set $U_\ell\gets \mpp{\ell}$.

\subsubsection{Superclustering}\label{sec superclustering st}

In this section we describe the superclustering step of phase $i$, for all $i\in [0,\ell]$. 

First, for each cluster $C\in \mathcal{P}_i$, its cluster center $r_C$ broadcasts its ID to all vertices in $C$. 
We then proceed to detecting popular clusters. Each vertex $u\in C$ notifies all its neighbors that $u$ belongs to the cluster $C$. Each vertex $v$ that belongs to some cluster $C'\in \mathcal{P}_i$ now knows the IDs of all clusters it is adjacent to. For each such neighboring cluster $C'$ that $u$ is adjacent to, the vertex $u$ sends to its own cluster center $r_C$ the message
$\langle r_{C'},u\rangle $, with the ID of the cluster center of $C'$. 
For each cluster $C\in \mathcal{P}_i$, its cluster center $r_C$ now aggregates all the information regarding neighboring clusters that resides in vertices of the cluster $C$. Each vertex that receives a message in this procedure, saves the cluster ID and sends it only if it has not already sent a message with the same ID. In any case, each vertex will send at most $deg_i$ messages. If after sending $deg_i$ messages a vertex receives messages regarding new cluster centers, it will discard these messages, i.e., they will never be sent.

Finally, when all communication in the algorithm terminates, i.e., all messages have either reached their destination, or have been discarded, each cluster center $r_C$ that received messages regarding at least $deg_i$ cluster centers adds $C$ to the set $W_i$ that will be returned by the algorithm. The pseudocode of the algorithm is provided below.

\begin{algorithm}
	\caption{Popular Clusters Detection}
	\label{alg Popular Clusters Detection}
	\begin{algorithmic}[1]
		\Statex \textbf{Input:} graph $G=(V,E)$, a set of clusters $\mathcal{P}_i$, parameter $deg_i$
		\Statex \textbf{Output:} a set $W_i$.
		
		\State\label{state 1} Each vertex $x\in V$ initializes a list of centers $x.\mathcal{LC}$ it learnt about as an empty list.
		
		\State\label{state 3} Each vertex $x\in C$ sends to all its neighbors in $G$ the message $\langle r_C \rangle$.
		
		\State\label{state 4} Each vertex $y\in C'$ that received messages $\langle r_{C}\rangle $, for $r_{C} \neq r_{C'}$, sends to its predecessor in $ T_{C'}$ the message $\langle r_{C},y\rangle $ 
		\For {each message $m$ a vertex $z\in C'$ receives}
		\If { $m$ informs $z$ of a center it did not know and $|z.\mathcal{LC}|<deg_i$ }
		\State\label{state 6} $z$ saves the message $m$ in $z.\mathcal{LC}$.
		\State\label{state 7} $z$ forwards the message $m$ to $r_{C'}$.
		\EndIf
		\EndFor
		\State\label{state 8} Each cluster center $r_C$ that has learnt about at least $deg_i$ other centers adds $C$ to the set $W_i$. 
	\end{algorithmic}
\end{algorithm}

In the following lemma, we prove the correctness of Algorithm \ref{alg Popular Clusters Detection}.

\begin{lemma}\label{lemma wi is set of popular}
	The set returned by Algorithm \ref{alg Popular Clusters Detection} is the set of popular clusters $W_i$. 
	Moreover, when the algorithm terminates, each center $r_{C'}$ that did not add $C'$ to $W_i$, knows the identities of all centers of clusters $C\in \Gamma_{\mathcal{P}_i}(C')$, and for every such $C$ it knows a vertex $y\in C'$ such that there is an edge $(y,x)\in E$, where $x\in C$. 
\end{lemma}

\begin{proof}
	Consider a cluster $C\in \mathcal{P}_i$. After the execution of \cref{state 3} of the algorithm, for each cluster $C'\in \Gamma_{\mathcal{P}_i}(C)$, there is a vertex $y\in C$ such that $y$ knows the $ID$ of $r_{C'}$. 	
	Therefore, in \cref{state 4} of the algorithm, at least
	${\min \{|\Gamma_{\mathcal{P}_i}(C)|, deg_i\}}$ messages are sent from vertices in $C$ to $r_C$. Note that a message from a vertex $v$ to $r_C$ is discarded only if $v$ has already sent $deg_i$ messages to $r_C$. It follows that by \cref{state 8} the cluster center $r_C$ knows 	${\min \{|\Gamma_{\mathcal{P}_i}(C)|, deg_i\}}$ IDs of other cluster centers.

	Let $C\in W_i$, i.e., $C$ is popular and $|\Gamma_{\mathcal{P}_i}(C)|\geq deg_i$. So $r_C$ received at least $deg_i$ messages and joined $W_i$. 
	Let $v$ be a vertex that has joined $W_i$. Then $v$ is a center of a cluster in $\mathcal{P}_i$, and it receives messages regarding at least $deg_i$ neighboring clusters. Therefore it is popular, i.e., it is in $W_i$.

	For the second assertion of the lemma, consider a center $r_C'$ that did not join $W_i$. From the first assertion of the lemma, we conclude that it is not popular, i.e., $|\Gamma_{\mathcal{P}_i}(C')|< deg_i$. Therefore, 
	it received messages regarding at least ${\min \{|\Gamma_{\mathcal{P}_i}(C)|, deg_i\}}=|\Gamma_{\mathcal{P}_i}(C')|$ other cluster centers, i.e., $r_{C'}$ has received all messages that were sent to it, and thus it knows the identities of all centers of clusters in $\Gamma_{\mathcal{P}_i}(C')$. Moreover, each message that was sent to it contains both an $ID$ of a cluster center $r_{C}$, and an $ID$ of a vertex $y\in C'$ that received a message $m= \langle r_{C} \rangle$ in \cref{state 3}. Since $y$ received the message $m$ in \cref{state 3}, we conclude that there is a vertex $x\in C$ that sent $m$ to $y$, thus $(x,y)\in E$. 
\end{proof}

Next, we select a subset of the popular clusters to grow superclusters around. 
We construct the virtual popular cluster graph $G'_i=(V',E')$, where $V'=\mathcal{P}_i$ and $E'$ contains edges from each popular cluster to its neighboring clusters (whether they are popular or not). Define $\delta =2/\rho$. We simulate the algorithm of Schneider et al. and Kuhn et at. \cite{sew,KuhnMW18} on the graph $G'_i$ to construct the subset $Q_i$, a $(3,\delta)$-ruling set for $W_i$, $Q_i\subseteq W_i$. The algorithm requires $O(n^{\rho})$ time. Details of this simulation can be found in Appendix \ref{apend rs}.

A BFS exploration is then simulated on $G'_i$ from all supervertices of $Q_i$ to depth $\delta$, and creates superclusters as in Section \ref{sec superclustering fast}. Define the sequence $R_0,\dots,R_\ell$ as in \cref{eq define ri gen}.
This concludes the description of the superclustering step. 

The following lemmas summarize the properties of the new superclusters.

\begin{lemma}\label{lemma single tree st}
	Let $i\in [0,\ell]$ and let $C$ be a cluster of $\mathcal{P}_i$. 
	At the beginning of phase $i$, the spanner $H$ contains a spanning tree $T_{C}$
	such that for every vertex $u\in V(C)$, there is a path in $T_{{C}}$ from $u$ to the cluster center $r_{C}$, of length at most $R_{i}$, that contains only vertices from $V({C})$. 
\end{lemma}

\begin{lemma}
	\label{lemma popular are clustered st} For every phase $i\in [0,\ell-1]$, all popular clusters in $\mathcal{P}_i$ are superclustered into clusters of $\mathcal{P}_{i+1}. $
\end{lemma} 

Their proofs are analogous to the proofs of Lemmas \ref{lemma popular are clustered fst} and \ref{lemma single tree}, respectively, and are thus omitted. Observe that Lemma \ref{lemma single tree st} implies that:

\begin{equation}
\label{eq ri bouns pi st}
Rad(\mathcal{P}_i)\leq R_i.
\end{equation}

\subsubsection{Interconnection}\label{sec intercon st}

We now discuss the details of the execution of the interconnection step of a phase $i\in [ 0, \ell]$.

Let $C\in U_i$. 
By Lemmas \ref{lemma popular are clustered st} and \ref{lemma wi is set of popular}, the cluster $C$ is not popular and for each neighboring cluster $C'\in \Gamma_{\mathcal{P}_i}(C)$, its center $r_C$ knows a vertex $v\in C$ such that there is an edge $(v,u)$ with $u\in C'$. 
In the interconnection step of phase $i$, for every cluster $C'\in \Gamma_{\mathcal{P}_i}(C)$, the cluster center $r_C$ will broadcasts to all vertices in $C$ the message $\langle r_{C'},v\rangle $. When the vertex $v$ receives this message, it will add the edge $(v,u)$ to the spanner $H$. If $v$ has multiple neighbors that belong to $C'$, it will arbitrarily choose one of them to add an edge to. This concludes the description of the interconnection step of phase $i$.

As in Section \ref{sec intercon fast}, denote by $U^{(i)}$ the union of all sets $U_0,U_1,\dots,U_i$, i.e., $U^{(i)} = \bigcup_{j=0}^iU_j$. Observe that like in the construction of Section \ref{sec polylog time spanner construction},   the set $U^{(\ell)}$ is a partition of $V$.

\subsection{Analysis of the Construction}\label{sec st Analysis of the Construction}

In this section, we analyze the parameters of the resulting spanner. Observe that Lemma \ref{eq explicit ri gen} and \cref{eq explicit Ri gen} are also applicable to the current construction. It follows that:

\begin{equation}
\label{eq explicit Ri st}
\begin{array}{cllll}
R_i 
&\leq &\frac{1}{2}\cdot (2\delta+1)^{i}

&= &\frac{1}{2}\cdot (4/\rho+1)^{i}.

\end{array}
\end{equation}

\subsubsection{Analysis of the Stretch}\label{sec st Analysis of the Stretch}
In this section, we analyze the stretch of the resulting spanner. Consider an edge $(u,v)\in E$. Since $U^{(\ell)}$ is a partition of $V$, both $u$ and $v$ belong to clusters of $U^{(\ell)}$. The following two lemmas bound the stretch of the edge $(u,v)$ in the spanner, in the case where $u,v$ belong to the same cluster in $U^{(\ell)}$, and in the case where $u,v$ belong to different clusters of $U^{(\ell)}$, respectively.

\begin{lemma}\label{lemma stretch st1}
	Let $(u,v)$ be an edge in the original graph $G$, such that $u,v$  belong to the same clusters of $U^{(\ell)}$. Then,
	$$	d_H(u,v) \leq 2R_\ell. $$
\end{lemma}	

The proof of the lemma is analogous to the proof of Lemma \ref{lemma stretch fa1}, thus it is omitted.

\begin{lemma}\label{lemma stretch st2}
		Let $(u,v)$ be an edge in the original graph $G$, such that $u,v$  belong to different clusters of $U^{(\ell)}$. Then,
		$$	d_H(u,v) \leq 4R_\ell+1. $$	
\end{lemma}	

\begin{figure}
	\begin{center}
		\includegraphics[scale=0.2]{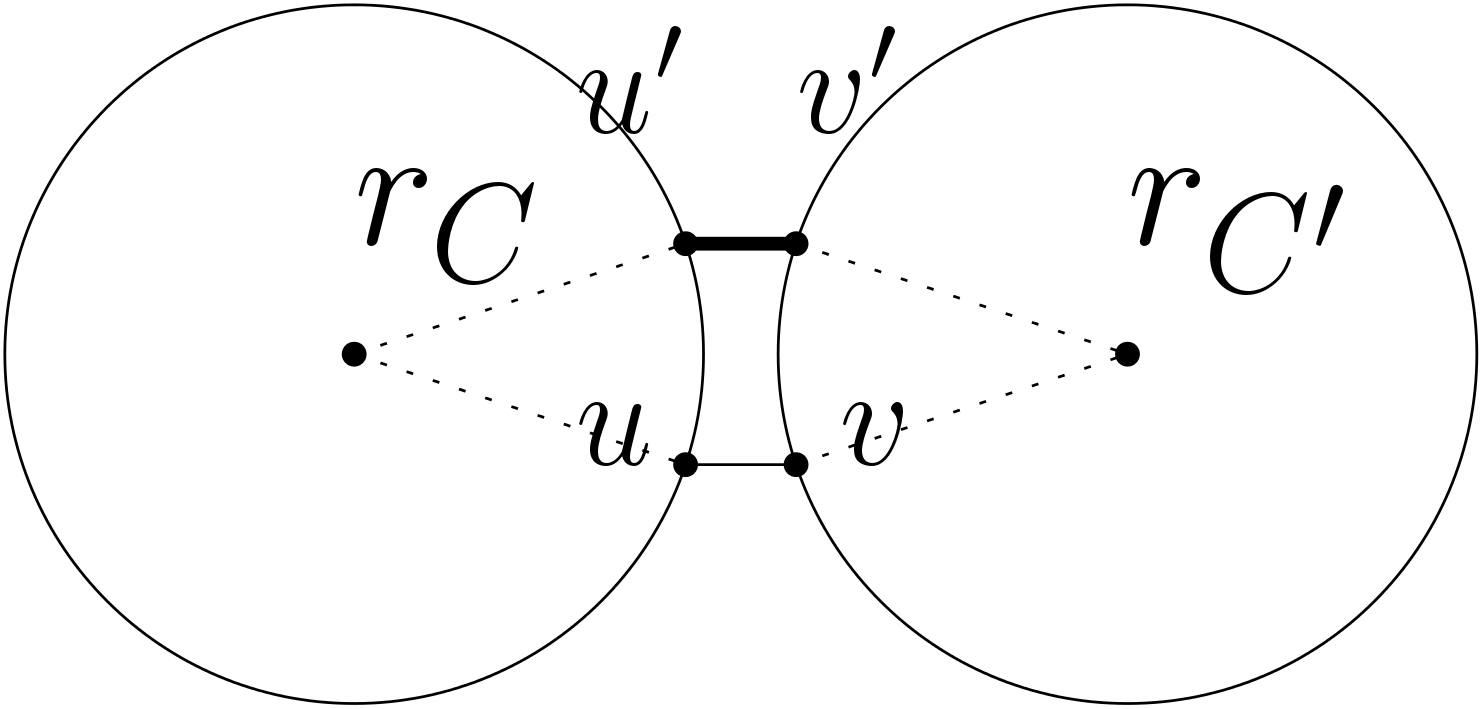}
		\caption{The path from $u$ to $v$ in $H$, if $u\in C$ and $v\in C'$. The line between $u,v$ represents the original $(u,v)$ edge from $G$. The dotted lines represent the paths in $H$ between the vertices $u,u',v,v'$ and their respective cluster centers. The thic line represents the edge $(u',v')$ that belongs to $G$ and to $H$.  }
		\label{fig not same_cluster}
	\end{center}
	
\end{figure}

\begin{proof}
	Let $C$ be the cluster such that $C\in U_i$ and $u\in C$. 
	Let $C''$ be the cluster such that $C''\in U_j$ and $v\in C''$. Assume w.l.o.g. that $i\leq j$. 	
	Let $C'$ be the cluster such that $C'\in \mathcal{P}_i$ and $v\in C'$. (Observe that if $i=j$, then $C''=C'$.) 
	
	Since  $C\in U_i$, in the interconnection step of phase $i$, the center $r_C$ of the cluster $C$ added edges to all its neighboring clusters. Specifically, an edge $(u',v')$ with $u'\in C$ and $v'\in C'$ was added to the spanner $H$. 
	
	By Lemma \ref{lemma single tree st}, there are paths in $H$ from $u,u',v,v'$ to their respective cluster centers $r_C,r_{C'}$ of length at most $R_i$. It follows that there is a path in $H$ between $u,v$ of length at most $4R_i+1$. Since $i\leq \ell$, we have 
	$$	d_H(u,v) \leq 4R_\ell+1. $$
\end{proof}

As a corollary to Lemmas \ref{lemma stretch st1} and \ref{lemma stretch st2}, we have:

\begin{corollary}\label{coro stretch st}
	For every edge $(u,v)\in E$ it holds that 
	\begin{equation}
	\label{eq stretch st}
	d_H(u,v) \leq 4R_\ell+1.
	\end{equation}
\end{corollary}

We will now derive an explicit expression of the stretch. 
Recall that 
$\ell= \floor*{{\log \kappa\rho}}+ \lceil\frac{\kappa+1}{\kappa\rho}\rceil$.
By \cref{eq explicit Ri st,eq stretch st}, it follows that for every edge $(u,v)\in E$, we have

\begin{equation}
\begin{array}{clllll}

d_H(u,v)&\leq& 4R_\ell+1\\
&\leq& 4(\frac{1}{2}\cdot (4/\rho+1)^{\ell})+1\\

&\leq& {2}\cdot (4/\rho+1)^{\ell} +1

\end{array}
\end{equation}

Therefore, for every pair of vertices $x,y\in V$, the distance between $x,y$ in $H$ satisfies:

\begin{equation}\label{eq stretch st final}
d_H(x,y)\leq
\left(	 {2}\cdot (4/\rho+1)^{\floor*{{\log \kappa\rho}}+ \lceil\frac{\kappa+1}{\kappa\rho}\rceil} +1\right)
\cdot d_G(x,y).		
\end{equation}

\subsubsection{Analysis of the Number of Edges}\label{sec st Analysis of the Number of Edges}

	In this section, we analyze the size of the spanner $H$. As in Section \ref{sec polylog time spanner construction}, we carefully examine the edges added by all phases of the algorithm, and charge every edge to a single vertex.

	Observe that $H$ contains two types of edges, the \textit{superclustering edges}, and the \textit{interconnection edges}. In this algorithm, we will charge each edge added in phase $i$ to a center of a cluster in $\mpi$. 
		
	As in Section \ref{sec polylog time spanner construction}, a superclustering edge that is added in a phase $i$ is an edge that connects a cluster $C\in \mathcal{P}_i\setminus Q_i$ to its predecessor in the BFS forest $F_i$. We will charge this edge to the center $r_{C}$ of the cluster $C$. See Section \ref{sec fa analysis size} for a detailed explanation and Figure \ref{fig superclustering charge} for an illustration.

An interconnection edge added in phase $i$ is an edge added to the spanner $H$ by a vertex that belongs to a cluster $C\in U_i$. In the current algorithm, a vertex $v$ will add an interconnection edge only if it received a message from its cluster center instructing it to do so. We will charge each interconnection edge to the cluster center that instructed $v$ to add the edge to the spanner $H$. For example, if a cluster $C\in U_i$ adds to the spanner $H$ edges $e_1,e_2,\dots,e_j$, for some $j$, in the interconnection step to clusters $C_1,C_2,\dots,C_j$, then the center $r_C$ of $C$ is charged for the $j$ edges ($e_1,e_2,\dots,e_j$). Observe that the cluster center $r_C$ will never be a center of a cluster in future phases. See Figure \ref{fig inter charge} for an illustration.

The following lemma shows essentially that a vertex $v\in V$ is charged at most once throughout the algorithm. It is either charged for a single edge when it is superclustered into another cluster, or it is a cluster center of a cluster in $U_i$, and then it is charged for less than $deg_i$ edges exactly once.

\begin{lemma}
	\label{lemma one charge}
	Each vertex $v\in V$ is charged for adding edges to $H$ in at most one phase of the algorithm.
\end{lemma}

\begin{proof}
	Let $v$ be a vertex, and let $i$ be the first phase of the algorithm in which $v$ is charged for an edge. Observe that in each phase of the algorithm, only cluster centers are charged for edges. Therefore, the vertex $v$ is a center of a cluster $C\in \mathcal{P}_i$ in phase $i$. 
	The vertex $v$ is charged either for a superclustering edge, or for interconnection edges. 
	
	\textbf{Case 1:} $v$ is charged for a superclustering edge. Then, the cluster $C$ was superclustered in phase $i$. By definition it does not belong to $U_i$, thus it is not charged for any interconnection edges in this phase.   
	Moreover, $v$ will not be a center of a cluster in future phases, thus it will not be charged for any edges in future phases.
	
	\textbf{Case 2:} $v$ is charged for interconnection edges.
	By definition, $C\in U_i$. 
	Observe that $v$ will not be a cluster center in future phases, and so it is charged for edges only in phase $i$. 
	
\end{proof}

Observe that by Lemma \ref{lemma one charge}, a vertex $v\in V$ is charged for at most one superclustering edge throughout the entire algorithm. Therefore, the superclustering steps of all phases $i\in [0,\ell-1]$ contribute at most $n$ edges to the spanner $H$, \textit{combined}.

We will now analyze the number of interconnection edges in the spanner $H$.
By Lemma \ref{lemma popular are clustered st}, if $C$ belongs to $U_i$, it has less than $deg_i$ neighboring clusters. Thus, its center will be charged for less than $deg_i$ edges. It is left to provide an upper bound on the size of the collections $U_i$, for all $i\in[0,\ell]$.
The superclustering step of phase $i$ partitions the set $\mathcal{P}_i$ into two disjoint sets: the set of clusters that are superclustered, and the set of clusters that are not superclustered, i.e., $U_i$.
In the following lemma, we use the size of $\mathcal{P}_{i+1}$ to upper bound the size of $U_i$.

\begin{lemma}\label{lemma bound size ui}
	For all phases $i\in[0,\ell-1]$, the size of the set $U_i$ is at most:
	$$|U_i| \leq |\mathcal{P}_i|-|Q_i|\cdot (deg_i+1) \leq  |\mathcal{P}_i|-|\mathcal{P}_{i+1}|\cdot (deg_i+1).$$
\end{lemma}

\begin{proof}
	Let $i\in[0,\ell-1]$, and let $\mathcal{P}_i$ be the set of clusters in phase $i$.
	In the superclustering step of phase $i$, the clusters of $Q_i$ have been chosen to grow larger clusters around them. These new superclusters are the clusters of $\mathcal{P}_{i+1}$, and so $|\mathcal{P}_{i+1}|= |Q_i|$. 
	
	By Lemma \ref{lemma wi is set of popular}, the set $W_i$ is the set of popular clusters. Since $Q_i$ is a subset of $W_i$, we know that all clusters in $Q_i$ are popular. Also, the set $Q_i$ is a $(3,2/\rho)$-ruling set for the set $W_i$ in the popular-cluster graph $G'_i$. 
	Therefore, for every pair of distinct clusters $C,C'\in Q_i$, we have $\Gamma_{\mathcal{P}_i}(C)\cap \Gamma_{\mathcal{P}_i}(C') = \emptyset$. It follows that the BFS exploration that originated from each cluster $C\in Q_i$ detects at least the clusters in $\Gamma_{\mathcal{P}_i}(C)$. Hence each cluster $\widehat{C}\in \mathcal{P}_{i+1}$ contains at least $deg_i$ clusters from $\Gamma_{\mathcal{P}_i}(C)$, and the cluster $C$ itself. 
	
	Thus, the size of the set $U_i$ of clusters from $\mathcal{P}_i$ that have not been superclustered in phase $i$ is at most
	$$|U_i|\leq |\mathcal{P}_i|-|Q_{i}|\cdot (deg_i+1) = |\mathcal{P}_i|-|\mathcal{P}_{i+1}|\cdot (deg_i+1).$$
\end{proof}

Next, we bound the number of interconnection edges added by all phases other than the concluding phase. Note that interconnection edges in phase $i$ are added to the spanner $H$ by clusters in $U_i$.

\begin{lemma}
	\label{lemma bound interconect}
	The number of interconnection edges added to the spanner $H$ by all phases $0,1,\dots,\ell-1$ is at most 
	$$|\mathcal{P}_0|\cdot deg_0-|\mathcal{P}_{\ell}|\cdot (deg^2_{\ell-1}+deg_{\ell-1}).$$
\end{lemma}

\begin{proof}
	
	We know that the number of edges added by the interconnection steps of each phase $i\in[0,\ell-1]$, is less than $|U_i|\cdot deg_i$.
	
	By Lemma \ref{lemma bound size ui}, the number of edges added by the interconnection steps of all phases $i\in [0,\ell-1]$ is smaller than:

	\begin{equation}
	\label{eq intercon overall contribution}
	\begin{array}{rlllll}
	&
	\sum_{i=0}^{\ell-1} |U_i|\cdot deg_i
	
\\

	\leq&
	\sum_{i=0}^{\ell-1}\left( |\mathcal{P}_i|-|\mathcal{P}_{i+1}|\cdot (deg_i+1)\right)\cdot deg_i
	\\
	
	\leq&
	|U_0|\cdot deg_0
	-|\mathcal{P}_{\ell}|\cdot (deg^2_{\ell-1}+ deg_{\ell-1})
	+
	\sum_{i=1}^{\ell-1} |\mathcal{P}_i|\cdot \left(
	deg_i-(deg_{i-1}^2+deg_i)
	\right).

	\end{array}
	\end{equation}
	
	Recall that in the exponential growth stage, i.e., phases $i\in [0,i_0]$, we have $deg_i = n^\frac{2^i}{\kappa}$. Also note that for the phase $i_0$ we have 

		\begin{equation*}
	\begin{array}{rlllllllll}
	deg_{i_0}
	&=& n^\frac{2^{\ize}}{\kappa} 
	&= &	n^\frac{2^{{\log \kappa\rho}-1}}{\kappa}
	&\geq& n^\frac{\kappa\rho}{2\kappa}
	&=& n^{\rho/2}.
	\end{array}
	\end{equation*}
	Recall also that for the fixed growth stage, i.e., phases $i\in [i_0+1,\ell-1]$ we have $deg_i = n^\rho$. 
	It follows that for every $i\in[0,\ell-1]$, we have that $deg_{i+1}\leq deg_i^2$. 
	Hence the number of edges added to the spanner by the interconnection steps of all phases $i\in[0,\ell-1]$ is at most

	\begin{equation}
	\label{eq intercon final contribution}
	\begin{array}{rlllll}
	|\mathcal{P}_0|\cdot deg_0-|\mathcal{P}_{\ell}|\cdot (deg^2_{\ell-1}+deg_{\ell-1}).
	\end{array}
	\end{equation}

\end{proof}
For the concluding phase $\ell$, we do not form superclusters and set $U_\ell=\mathcal{P}_\ell$. We will show now that the size of $U_\ell$ is at most $n^\rho$. The following three lemmas provide upper bounds on the size of the collections $\mathcal{P}_i$ in the exponential growth stage (along with the transition phase) and in the fixed growth stage, respectively.

\begin{lemma}
	\label{lemma bound size pi}
	For all $i\in [1,\ell]$ we have 
	$$|\mpi|\leq |P_{i-1}|\cdot (deg_{i-1})^{-1}.$$
\end{lemma}

\begin{proof}
	
	For every index $i\in [1,\ell-1]$, each cluster $C\in \mathcal{P}_{i+1}$ was constructed by the BFS exploration that originated from a cluster in $Q_{i}$. Therefore, we have $|\mathcal{P}_{i+1}|= |Q_{i}|$. The set $Q_{i}$ is a $(3,{2}/{\rho})$-ruling set for $W_i$ in $G'_i$. By Lemma \ref{lemma wi is set of popular}, all vertices in $W_i$ are popular cluster centers. Thus, for every $C\in W_i$, it holds that
	$ |\Gamma_{\mathcal{P}_i}(C)| \geq deg_i$. 
	Define $\widehat{\Gamma}_{\mpi}(C) = \{C\}\cup \Gamma_{\mpi}(C)$ i.e., the set of neighbors of $C$ as well as $C$ itself.
	Observe that $|\widehat{\Gamma}_{\mpi}(C)|\geq deg_i+1$

	The set $Q_i$ is $3$-separated, i.e., for every pair of distinct clusters $C,{C'}\in Q_i$ we have $ d_{G'_i}(C,{C'})\geq 3$.
	
	Thus, every pair of distinct clusters $C,C'\in Q_i$ we have $\widehat{\Gamma}_{\mathcal{P}_i}(C)\cap \widehat{\Gamma}_{\mathcal{P}_i}(C') = \emptyset$.
	
	It follows that:
	\begin{equation*}
	|P_{i+1}|\leq {|P_{i}|}\cdot ({deg_i+1})^{-1}
	\end{equation*}
\end{proof}

\begin{lemma}\label{inter_rt_lm2}

		For $i\in [0, i_0+1= \floor*{\log(\kappa\rho)}+1]$, we have 
		$$| {P}_i| \leq n^{1-\frac{2^i-1}{\kappa}}.$$

\end{lemma}

\begin{lemma}\label{inter_rt_lm5}
	For $i_0+1\leq i \leq \ell$, it holds that
	$$| {P}_i| \quad \leq\quad n^{1+\frac{1}{\kappa}-(i-i_0)\rho}. $$
\end{lemma}

The proofs of Lemmas \ref{inter_rt_lm2} and \ref{inter_rt_lm5} are analogous to the proof of Lemmas  2.10 and 2.11 in \cite{ElkinMatar}, and are therefore deferred to Appendix \ref{append some proofs}.

Recall that $\ell= \floor*{{\log \kappa\rho}}+ \lceil\frac{\kappa+1}{\kappa\rho}\rceil-1$, and that $i_0= \floor*{{\log \kappa\rho}}$. By Lemma \ref{inter_rt_lm5}, the size of $\mathcal{P}_\ell$ is bounded by: 

\begin{equation}
\label{eq bound pl}
|\mathcal{P}_{\ell}|\leq n^{1+\frac{1}{\kappa}-(\ell-i_0)\rho}=n^{1+\frac{1}{\kappa}-(\lceil\frac{\kappa+1}{\kappa\rho}\rceil-1)\rho}
\leq n^\rho.
\end{equation}

Observe that \cref{eq bound pl} implies that in the concluding phase $\ell$, we add at most $|\mathcal{P}_\ell| \cdot (n^\rho-1)$ interconnection edges to $H$. 
Recall that $|\mathcal{P}_0|=n$, and that all superclustering steps add at most $n$ edges combined. By \cref{eq intercon final contribution}, we obtain that the size of the spanner $H$ is bounded by: 
\begin{equation}
\label{eq size}
\begin{array}{cllllll}
|H| 
&<& n+|\mathcal{P}_0|\cdot deg_0-|\mathcal{P}_{\ell}|\cdot (deg^2_{\ell-1}+deg_{\ell-1})+|\mathcal{P}_\ell| \cdot 
(n^\rho-1)\\
&=& n +n\cdot \nk -|\mathcal{P}_{\ell}|\cdot (n^{2\rho}+n^\rho)+|\mathcal{P}_\ell| \cdot 
(n^\rho-1)\\
&\leq& \nfrac+n
\end{array}
\end{equation}

\subsubsection{Analysis of the Running Time}\label{sec st Analysis of the Running Time}

	We begin by analyzing the running time of a single phase $i$. 
	
	\textbf{Superclustering.} The superclustering step of phase $i$ begins with detecting the popular clusters. 
	By \cref{eq ri bouns pi gen}, downcasting $m$ messages from the center of the cluster $C$ to all vertices in $C$ requires $O(m+R_i)$ time. Also, upcasting $m$ messages from vertices in the cluster $C$ to the center of $C$ requires $O(m+R_i)$ time. By Lemma \ref{lemma single tree}, one can transfer data within different clusters in parallel, without having messages from two distinct clusters interfering with one another. 
	Therefore, Algorithm \ref{alg Popular Clusters Detection} requires $O(deg_i\cdot R_i)$ time. 
	
	The algorithm of \cite{sew,KuhnMW18} computes a $(3,2/\rho)$-ruling set in a graph $G$ on $x$ vertices in $O(n^\rho)$ time. Since in the algorithm of \cite{sew,KuhnMW18} every vertex sends the same message on each round to all its neighbors, the algorithm applies to the Broadcast-CONGEST model. Therefore, the algorithm can be simulated on a supergraph, where the overhead is the maximum diameter of a supervertex in a simulated supergraph. Therefore computing a $(3,2/\rho)$ requires $O(R_i\cdot n^\rho)$ time. 
	
	The BFS exploration to depth $2/\rho$ in $G'_i$ requires $O(R_i\cdot 2/\rho)$ time. 
	
	\textbf{Interconnection.} 
	In the interconnection step of phase $i$, each center $r_C$ of a cluster $C\in \mathcal{P}_i$ broadcasts less than $deg_i$ massages to all vertices in its cluster.
	Each vertex that receives a relevant massage, adds a single edge to the spanner $H$. Therefore, the interconnection step of phase $i$ requires $O(R_i\cdot deg_i)$ time. 
	
	Recall that for all $i\in [0,\ell]$, we have $deg_i\leq n^\rho$.
	Therefore, the running time of a single phase of the algorithm is $O(R_i\cdot n^\rho)$.

	By \ref{eq explicit Ri st}, and since $\ell ={\floor*{{\log \kappa\rho}}+ \lceil\frac{\kappa+1}{\kappa\rho}\rceil}-1$, we have that the running time of the entire algorithm is 
	\begin{equation}
	\label{eq alg st running time}
	\begin{array}{cllll}
	O\left(
	
	\sum^{\ell}_{j=0} O(R_j\cdot n^\rho)
	\right) =

	O\left(
	n^\rho\cdot
	\sum^{\ell}_{j=0} (\frac{1}{2}\cdot (4/\rho+1)^{j})
	\right) =\\

	O\left(
	n^\rho\cdot
	\left[
	\frac{(4/\rho+1)^{\ell+1}-1}
	{(4/\rho+1)-1}
	\right]
	\right) =
	
	O\left(
	n^\rho\cdot
	(4/\rho+1)^{\floor*{{\log \kappa\rho}}+ \lceil\frac{\kappa+1}{\kappa\rho}\rceil }
	
	\right) 
	\\

	\end{array}
	\end{equation}

As a corollary to \cref{eq stretch st final,eq size,eq alg st running time},

\begin{corollary}
	\label{corollary st construction}
	For any  parameters $\kappa\geq 2$, and $\frac{1}{\kappa}\leq \rho<\frac{1}{2}$, and any $n$-vertex unweighted undirected graph $G=(V,E)$, our algorithm constructs a $t $-spanner with
	$ \nfrac $	edges,
	in 
	$ O\left(n^\rho\cdot t \right)$
	deterministic time in the CONGEST model, where $t =
	\left(\frac{4}{\rho}+1\right)^{ {\log \kappa\rho}+\frac{1}{\rho}+O(1)}
	$. 
	
\end{corollary}	
	

	\bibliographystyle{alpha}
	\bibliography{cite}
	
\begin{appendix}

	\section{Some Proofs}
	\label{append some proofs}
	
	\textbf{Lemma \ref{inter_rt_lm2}} \textit{
		For $i\in [0, i_0+1= \floor*{\log(\kappa\rho)}+1]$, we have 
	$$| \mpi| \leq n^{1-\frac{2^i-1}{\kappa}}.$$
	}

	\begin{proof}
		We will prove the lemma by induction on the index of the phase $i$.

		For $i=0$, the right-hand side is equal to $n$ and therefore the claim is trivial.

		Assume the claim holds for $i\in [0,i_0]$  and prove it for $i+1$. By Lemma  \ref{lemma bound size pi},
		and by the induction hypothesis, we have that $|\mpp{i+1}|\leq \mpi\cdot (deg_{i})^{-1} $.
		
		\begin{equation}
		|\mpp{i+1}|\leq \mpi\cdot (deg_{i})^{-1} \leq n^{1-\frac{2^i-1}{\kappa}}\cdot n^\frac{-2^i}{\kappa} = 
		 n^{1-\frac{2^{i+1}-1}{\kappa}}
		\end{equation}
	\end{proof}
	
	Observe that Lemma \ref{inter_rt_lm2} implies that for $i_0+1$ we have:

	\begin{equation}\label{Bound P i 0+1}
	| \mpp{i_0+1}|	\leq n^{1-\frac{2^{\ize+1}-1}{\kappa}}\leq 
	 n^{1+\frac{1}{\kappa}-\frac{\kappa\rho}{\kappa}}\leq n^{1+\frac{1}{\kappa}-\rho}.
	\end{equation}

	\textbf{Lemma \ref{inter_rt_lm5}} \textit{	For $i_0+1\leq i \leq \ell$, it holds that
		$$| {P}_i| \quad \leq\quad n^{1+\frac{1}{\kappa}-(i-i_0)\rho}. $$}
	
	\begin{proof}

		The proof is by induction on the index of the phase $i$. 
		
		For the base case, by \cref{Bound P i 0+1}, we have 

		\begin{equation*}
		\begin{array}{ccccc}
			| \mpp{i_0+1}|	\leq n^{1+\frac{1}{\kappa}-\rho}= n^{1+\frac{1}{\kappa}-(i_0+1-i_0)\rho}.
		\end{array}
		\end{equation*}
		
		Assume the claim holds for $i\in [i_0+1,\ell-1]$  and prove it for $i+1$. By Lemma  \ref{lemma bound size pi},
		and the induction hypothesis we have that $|\mpp{i+1}|\leq \mpi\cdot (deg_{i})^{-1} $.
		
		Together with the induction hypothesis, this implies that 
		\begin{equation*}
		|\mpp{i+1}|\leq |\mpi|\cdot (deg_{i})^{-1}
		\leq n^{1+\frac{1}{\kappa}-(i-i_0)\rho}\cdot n^{-\rho} 
		= n^{1+\frac{1}{\kappa}-(i+1-i_0)\rho}
		\end{equation*}
	\end{proof}

\section{Ruling set}
\label{apend rs}

An algorithm for the construction of a $(c+1,cq)$-ruling set in the \congestmo\ is devised in \cite{sew,KuhnMW18}. The following theorem is derived from their result.

\begin{theorem}\label{sew_thorem}
	Given a graph $G=(V,E)$ in which each vertex $v\in V$ has a unique ID in the range $[n]$, a set $A\subseteq V$, and a parameters $q$, a $(3,2q)$-ruling set for $A$ can be built in $O(q \cdot n^\frac{1}{q})$ time, in the \congestmo. 
\end{theorem} 

A pseudocode of the algorithm is provided below.

\begin{algorithm}
	\caption{CONGEST-Ruling-Set}\label{alg rs_congest}
	\begin{algorithmic}[1]

		\State $\boldsymbol{Input}$ Graph $G=(V,E)$, $|V|=n$, vertices $A$ with IDs from $[a,b]$, parameters $q,c$.
		\State $\boldsymbol{Output}$ Ruling set $RS$
		\If {$b-a\leq 1$}		\Comment{$A$ is a singleton}
		\State $RS \gets A $

		\Else
		\State $t \gets n^\frac{1}{c}$	\Comment{number of sets}
		\State $r \gets \frac{(b-a)}{t}$	\Comment{an upper bound on the number of elements in each set}
		
		\For {$l\gets 0 \ to \ t-1$ in parallel}
		\State $A_l \gets \{ v \in A\ | ID(v)\in [a+l\cdot r,a+(l+1)\cdot r -1] \} $
		\State $RS_l \gets \textit{CONGEST-Ruling-Set}(G,A_l,q,c) $
		\EndFor
		\State $RS\gets \emptyset$ 
		\For {$l=0\ to\ t-1$} 				\Comment{computed sequentially }
		\State $RS\gets RS\cup RS_l$
		\State all vertices $v\in RS_l$ broadcast a \textit{knock-out} message to depth $q$.
		\State each vertex $u\in RS_{l'} $ for $ l<l'<t$ that receives a \textit{knock-out} message (from some source, not necessarily in $RS_l$) is removed from $ RS_{l'} $ 
		\EndFor	
		\EndIf

	\end{algorithmic}
	
\end{algorithm}

Note that throughout the entire algorithm, there is only one type of message that is being sent on graph edges, i.e., the knock out message.

We simulate on the graph $G$ the execution of this algorithms to construct ruling sets in a virtual supergraph $G'$.

\begin{theorem}\label{theo simu}
	Let $G=(V,E)$ be an $n$-vertex graph 
	in which each vertex has a unique ID in the range $[n]$. Let $H\subseteq E$ be a set of edges. Let $G'=(V',E')$ be  a virtual supergraph, where:
	\begin{enumerate}
		\item Each supervertex is a set of vertices $C$ with a designated center $r_C$.
		\item For every cluster $C$, the set $H$ contains a tree $T_C$ that contains all and only vertices of $C$.
		\item For very cluster $C\in V'$, and for every vertex $u\in C$, the distance between $u$ and the center of the cluster $C$ is at most $R$. 
	\end{enumerate} 
	 Let  $A\subseteq V$, and let $q$ be a parameters.
	 Then,  a $(3,2q)$-ruling set for $A$ in $G'$ can be built in $O(R\cdot q \cdot n^\frac{1}{q})$ time, in the \congestmo. 
\end{theorem}

\begin{proof}
	We will show that each communication round of Algorithm \ref{alg rs_congest} can be simulated in $G$ by $2R$ communication rounds.

	Our simulation algorithm will run the Algorithm \ref{alg rs_congest}, step by step.
	Every supervertex (cluster) $C\in V'$ will be simulated by its cluster center $r_C$. The ID of the cluster is set to be the ID of the cluster center. For every cluster $C\in V'$ that needs to send a knock-out message, its cluster center $r_C$ will broadcast this message to all vertices in $C$. This requires $R$ time.  Then, each vertex $v\in C$ send a knock out message to all vertices $u\in C'$ such that $(u,v)\in E$. In other words, $v$ sends the message to all vertices that belong to neighboring clusters of $C$. Each vertex $u\in C'$, for some $C'\in V'$,  that received the knock out message will now deliver it to the center $r_{C'}$ of the cluster $C'$. 
	
\end{proof}
 
The algorithm of \cite{AwerbuchGLP89} for constructing ruling sets can be viewed as a special case of this algorithm, when $q= {\log n}$. 

\end{appendix}

\end{document}